\documentclass[
pra,aps,
reprint,
a4paper,
superscriptaddress,
longbibliography,
floatfix
]{revtex4-2}
\usepackage{graphicx}
\usepackage{epsfig}
\usepackage{amsfonts}
\usepackage{amsmath}
\usepackage{amssymb}
\usepackage{color}
\usepackage{amsthm}% in order to use a black box at the end of proofs
\usepackage[colorlinks=true,linkcolor=blue,citecolor=magenta,urlcolor=blue]{hyperref}

\newtheorem{theorem}{Theorem}
\newtheorem{colloraly}{Colloraly}

\newtheorem{remark}{Remark}
\newtheorem{Definition}{Definition}
\newtheorem{Example}{Example}

\begin{document}

%%%%%%%%%%%%%%%%%%%%%%%%%%%%%%%%%%%%%%%%%%%%%%%%%%%%%%%%%%%%%%%%%%%

\title{Converting contextuality into nonlocality}

%%%%%%%%%%%%%%%%%%%%%%%%%%%%%%%%%%%%%%%%%%%%%%%%%%%%%%%%%%%%%%%%%%%

\author{Ad\'an Cabello}
\email{adan@us.es}
\affiliation{Departamento de F\'{\i}sica Aplicada II, Universidad de	Sevilla, E-41012 Sevilla, Spain}
\affiliation{Instituto Carlos~I de F\'{\i}sica Te\'orica y Computacional, Universidad de
	Sevilla, E-41012 Sevilla, Spain}

%%%%%%%%%%%%%%%%%%%%%%%%%%%%%%%%%%%%%%%%%%%%%%%%%%%%%%%%%%%%%%%%%%%

\begin{abstract}
We introduce a general method which converts, in a unified way, any form of quantum contextuality, including any form of state-dependent contextuality, into a quantum violation of a bipartite Bell inequality. As an example, we apply the method to a quantum violation of the Klyachko-Can-Binicio\u{g}lu-Shumovsky inequality. 
\end{abstract}

% 5+8 pages, 2+8 figures
% Phys. Rev. Lett. 127, 070401 (2021)
% 10.1103/PhysRevLett.127.070401 

%%%%%%%%%%%%%%%%%%%%%%%%%%%%%%%%%%%%%%%%%%%%%%%%%%%%%%%%%%%%%%%%%%%

\maketitle

%%%%%%%%%%%%%%%%%%%%%%%%%%%%%%%%%%%%%%%%%%%%%%%%%%%%%%%%%%%%%%%%%%%

{\em Introduction.---}Nonlocal games \cite{CHTW04,BCPSW14} provide an intuitive understanding of where the advantage of quantum resources lies and a framework, used in computer science \cite{CL89}, to analyze quantum protocols. Contextuality is known to be a crucial resource for some forms of computation with quantum speed up \cite{AB09,Raussendorf13,HWVE14}. 
However, although some forms of contextuality can be converted into nonlocal games, there is no universal method for converting any form of contextuality into a nonlocal game. The aim of this Letter is to introduce a unified method that achieves this task.

The research on how contextuality can be converted into nonlocality started with the works of Stairs \cite{Stairs83} and Heywood and Redhead \cite{HR83}, extending the proofs of the Kochen-Specker (KS) theorem \cite{KS67} to bipartite scenarios with entanglement and has evolved in many ways in connection to extensions of the KS theorem \cite{BS90,CK06}, Bell inequalities \cite{Cabello01,GMS07,AGA12}, and nonlocal games \cite{Aravind02,GBT05,CMNSW07,CM14,MR16,ASDZ17}.

So far, the forms of contextuality that can be converted into nonlocality are (i)~Those forms of state-dependent contextuality (SD-C) corresponding to scenarios whose measurements can be distributed between two or more parties in such a way that each party has at least two incompatible measurements. 
(ii)~Those that are produced by KS sets \cite{KS67} (i.e., sets of rank-one projectors which do not admit a ``KS assignment,'' i.e., an assignment of $0$ or $1$ satisfying that two orthogonal projectors cannot both have assigned~$1$, and for every set of mutually orthogonal projectors summing the identity, one of them must be assigned~$1$) or by proofs of the KS theorem (e.g., \cite{Peres90,Mermin90}) that can be reduced to KS sets \cite{Peres91,KP95}. For methods of conversion, see, e.g., \cite{CHTW04,AGA12}. (iii)~Those produced by some particular state-independent contextuality (SI-C) sets \cite{YO12} (i.e., sets of projectors which produce noncontextual behaviors for any initial state) that are not KS sets. For methods of conversion, see \cite{CAB12}. (iv)~In addition, some constraint satisfaction problems and local no-hidden-variables proofs can be converted into nonlocal games \cite{CMNSW07,CM14,MR16,ASDZ17}. In each case, ``convert'' may mean a different thing.

Forms of contextuality that we do not know how to convert into nonlocality are those produced by sequentially measuring noncomposite systems initially prepared in specific states in contextuality scenarios which cannot be embedded in Bell scenarios. A particularly relevant example is the quantum violation of the Klyachko-Can-Binicio\u{g}lu-Shumovsky (KCBS) inequality \cite{KCBS08} with single qutrits. This is arguably the most fundamental form of quantum SD-C produced by noncomposite systems, as the KCBS inequality is the only nontrivial tight noncontextuality inequality \cite{AQB13} in the scenario with the smallest number of measurements in which qutrits produce contextuality (qutrits are the quantum systems of smallest dimension that produce contextuality \cite{KS67}), and because it plays a crucial role for understanding quantum contextuality 
\cite{CSW14,Cabello13,Cabello19}.

The aim of this Letter is to provide a general unified method capable of converting any form of SD-C or SI-C into bipartite nonlocality. The philosophy behind the method is guided by the recognition of the singular role of SI-C, as pointed out in, e.g., \cite{XYK21} (``we argue that a primitive entity of contextuality should embrace state-independence''). The method takes any set of measurements that provides SD-C and identifies the minimal extension of it that provides SI-C and then converts the SI-C into bipartite nonlocality preserving the gap between quantum and noncontextual theories in the SI-C (which becomes the gap between quantum and local theories).

First, we describe the method, which has three steps. Then, we apply the method to a quantum violation of the KCBS inequality. Finally, we provide an intuitive explanation of how it works and discuss its virtues and limitations.

%%%%%%%%%%%%%%%%%%%%%%%%%%%%%%%%%%%%%%%%%%%%%%%%%%%%%%%%%%%%%%%%%%%

{\em Method.---}An ideal measurement of an observable $A$ is a measurement of $A$ that yields the same outcome when repeated and does not disturb any compatible (i.e., jointly measurable) observable. A context is a set of ideal measurements of compatible observables. A scenario is characterized by a number of measurements, their outcomes, and relations of compatibility \cite{Cabello19}. In quantum theory, every ideal measurement is represented by the spectral projectors of a self-adjoint operator, and compatible observables correspond to commuting operators. 

A behavior for a scenario (i.e., a set of probability distributions for each of its contexts) is contextual if the probability distributions for each context cannot be obtained as the marginals of a global probability distribution on all observables. Otherwise, the behavior is noncontextual. Contextuality is detected by the violation of noncontextuality inequalities whose bounds are derived solely from the assumption of outcome noncontextuality \cite{KCBS08,Cabello08,BBCP09,YO12,KBLGC12}. Any quantum contextual behavior can be produced by a set of rank-one projectors $S=\{\Pi_1,\ldots,\Pi_n\}$ acting on a quantum state $|\psi\rangle$ in a Hilbert space of dimension $d\ge 3$ \cite{CSW14}. Given $S$, contexts are subsets of $S$ containing mutually commuting projectors.

%%%%%%%%%%%%%%%%%%%%%%%%%%%%%%%%%%%%%%%%%%%%%%%%%%%%%%%%%%%%%%%%%%%
% Fig. 1
%%%%%%%%%%%%%%%%%%%%%%%%%%%%%%%%%%%%%%%%%%%%%%%%%%%%%%%%%%%%%%%%%%%

\begin{figure}[t!] %\centering
	\includegraphics[width=.46\textwidth]{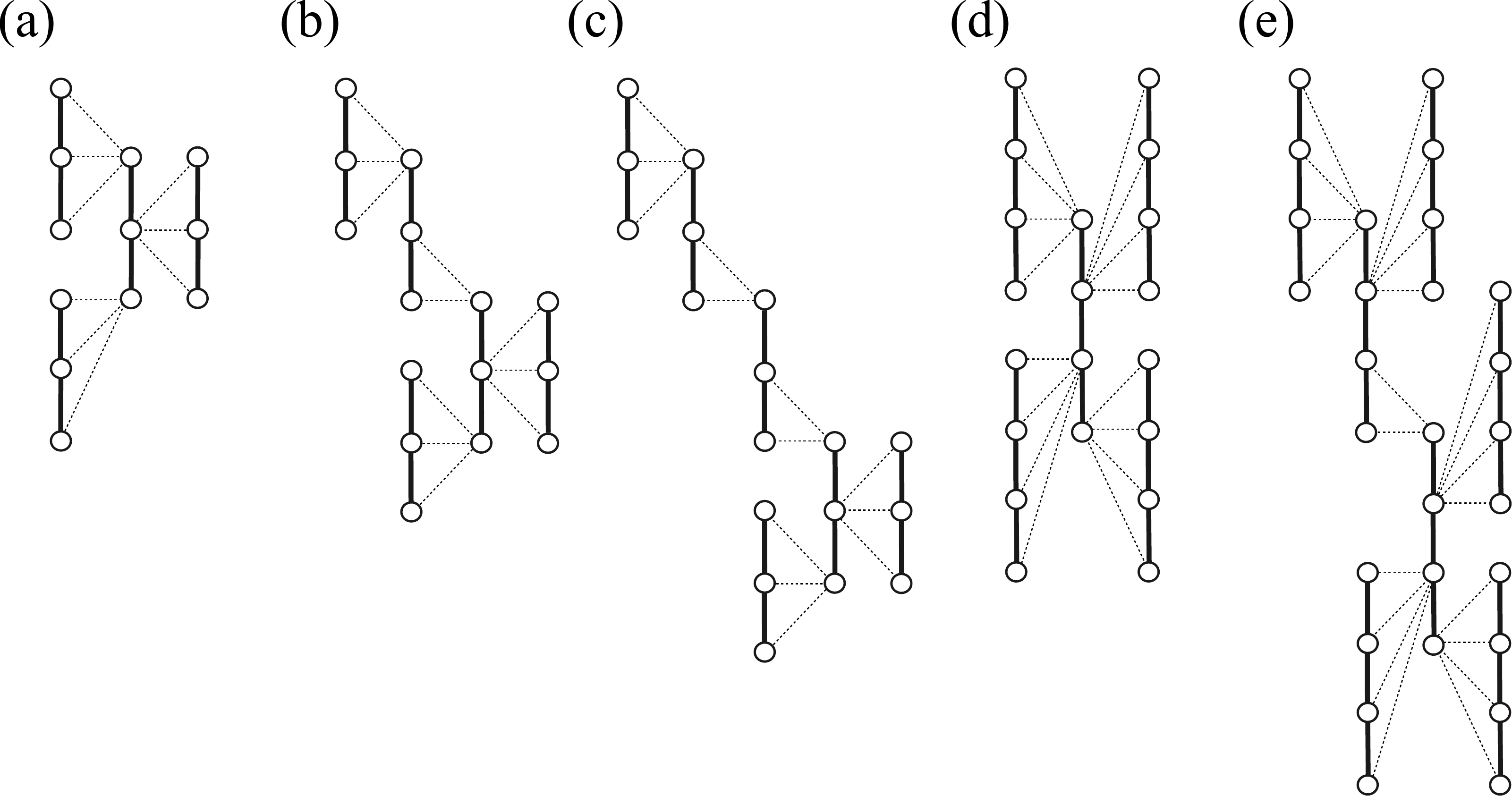}
	\caption{Every node represents a rank-one projector. A continuous vertical line between $d \ge 3$ nodes indicates that they are mutually orthogonal. Hence, in dimension $d$, in any KS assignment, one of them has to be assigned~$1$. A dashed line between two nodes indicates that there is a TIFS between (and including) them. Hence, in any KS assignment, both of them cannot be assigned~$1$. Construction to obtain a critical KS set in dimension $d\ge 3$ from $N \ge d+1$ disjoint bases: (a)~For $d=3$ and $N=d+1$. (b)~For $d=3$ and $N=d+2$. (c)~For $d=3$ and $N=d+3$. (d)~For $d=4$ and $N=d+1$. (e)~For $d=4$ and $N=d+2$. The construction works similarly for any $d\ge 3$ and $N \ge d+1$. In all cases, it is impossible to assign to the depicted nodes the values $0$ or $1$ satisfying that one of the $d$~nodes in each continuous vertical line must be~$1$, while nodes connected by a dashed line cannot both be~$1$. However, such an assignment is possible whenever we remove any of the depicted nodes. \label{fig1}}
\end{figure}

%%%%%%%%%%%%%%%%%%%%%%%%%%%%%%%%%%%%%%%%%%%%%%%%%%%%%%%%%%%%%%%%%%%

Step 1: A SI-C set is critical if by removing any of its elements the resulting set is not a SI-C set. A KS set is critical \cite{ZP93} if by removing any of its elements the resulting set is not a KS set. Here we show that every set $S$ producing SD-C can be extended into a critical SI-C set $S''= S \cup S'$. To prove this, we need the following result \cite{CG96,RRHPHH20}. In $d \ge 3$, given any two nonorthogonal rank-one projectors $\Pi_A$ and $\Pi_B$, there is a set of projectors $E$ such that, for any KS assignment~$f$, $f(\Pi_A) + f(\Pi_B) \le 1$. The set $\Pi_A \cup E \cup \Pi_B$ is called a true-implies-false set (TIFS) \cite{CPSS18}, definite prediction set \cite{CG96}, $01$-gadget \cite{RRHPHH20}, or Hardy-like proof \cite{XCG20}.
 
The construction of a critical SI-C set containing $S$ is as follows. Let $G$ be the graph of orthogonality of $S$. Let $N$ be the minimum number of disjoint bases that cover all the vertices of~$G$. If $S$ allows for SD-C, then $N \ge 3$ \cite{CSW14}. If $N < d+1$, then we add disjoint bases until the total of number of disjoint bases is $N+1$. Then, we use the construction shown, for $d=3$, in Fig.~\ref{fig1}(a) and, for $d=4$, in Fig.~\ref{fig1}(d), and which works similarly for any $d\ge 5$, based on creating TIFSs between some specific nodes. If $N>d+1$, then we use the construction shown, for different combinations of $d$ and $N$, in Figs.~\ref{fig1}(b), \ref{fig1}(c), or \ref{fig1}(e). In all cases, the resulting set is a critical KS set in dimension $d$ for the reasons explained in Fig.~\ref{fig1}.
If one removes any of the nodes in each of the constructions in Fig.~\ref{fig1}, then the resulting set admits a KS noncontextual assignment.
Some SI-C sets are not KS sets (e.g., \cite{YO12,BBC12,XCS15}). Hence, the resulting critical KS set could, in principle, not be a critical SI-C set. However, this problem can be solved by suitably choosing the extra nodes used for the TIFSs in Fig.~\ref{fig1} (see Appendices~\ref{sec2} and \ref{sec3}).

A minimal critical SI-C set is a critical SI-C set of minimum cardinality. The previous proof guarantees that critical SI-C sets containing $S$ exist. However, the method used in the proof does not guarantee that the resulting critical SI-C set is minimal. To obtain a minimal critical SI-C set $S''$ from $S$ we can use the following results. Let us call ${\cal G}$ the graph of orthogonality of $S''$ and let $d$ be the dimension of the Hilbert space. Necessary conditions for $S''$ to be a SI-C set are that the chromatic number of ${\cal G}$ satisfies $\chi ({\cal G}) > d$ \cite{Cabello11} and that the fractional chromatic number satisfies $\chi_f ({\cal G}) > d$ \cite{RH14,CKB15}. These conditions allow us to identify candidates to be minimal critical SI-C sets containing any given SD-C set. Then, we can use the necessary and sufficient condition for being a SI-C set \cite{CKB15} to check whether or not they are SI-C sets. This condition states that a set of rank-one projectors $S'' = \{\Pi_i,\ldots,\Pi_n\}$ is a SI-C set if and only if there are nonnegative numbers $w = (w_1, \ldots, w_n)$ and a number $0 \leq y < 1$ such that $\sum_{j \in \cal I} w_j \le y$ for all ${\cal I} $, where ${\cal I}$ is any independent set of ${\cal G}$, and $\sum_{i} w_i \Pi_i \ge \openone$.

In practice, finding a critical SI-C set containing $S$ is not a problem. However, proving that it has minimal cardinality may be difficult. See \cite{XCG20,CKP16} for examples of such proofs (see also Appendix~\ref{sec4}). Nevertheless, minimality is only required for elegance; to connect SD-C to nonlocality what matters is the criticality of the SI-C set. 

%%%%%%%%%%%%%%%%%%%%%%%%%%%%%%%%%%%%%%%%%%%%%%%%%%%%%%%%%%%%%%%%%%%

Step 2: As pointed out in \cite{CKB15}, the weights $w$ needed to guarantee that $S''$ is a SI-C set generate a noncontextuality inequality violated by any quantum state. The results in \cite{CSW14,Cabello16} allow us to express this inequality as
\begin{widetext}
	\begin{equation}
		\label{SI-C}
		\sum_{i \in(see  V({\cal G})} w_i P(\Pi_i =1) - \sum_{(i,j) \in E({\cal G})} \max(w_i, w_j) P(\Pi_i =1, \Pi_j =1) \stackrel{\mbox{\tiny{NCHV}}}{\leq} \alpha({\cal G},w),
	\end{equation}
\end{widetext}
where $P(\Pi_i =1, \Pi_j =1)$ is the probability of obtaining outcome $1$ in the measurement associated to $\Pi_i$ (which has possible outcomes $1$ and $0$) and also in the measurement associated to $\Pi_j$, $V({\cal G})$, and $E({\cal G})$ are the sets of vertices and edges of ${\cal G}$, respectively, $\alpha({\cal G},w)$ is the independence number of $({\cal G},w)$ [i.e., the graph in which weight $w_i$ is assigned to each $i \in V({\cal G})$], and NCHV stands for noncontextual hidden-variable theories. The independence number of a (weighted) graph is the cardinality of its largest set of vertices (taking their weights into account) such that no two are adjacent.

%%%%%%%%%%%%%%%%%%%%%%%%%%%%%%%%%%%%%%%%%%%%%%%%%%%%%%%%%%%%%%%%%%%

Step 3: This step has two ingredients. One is the following method, introduced in \cite{Stairs83,HR83} and used extensively since then to embed a KS set in a bipartite Bell scenario. In each run of the experiment, we prepare a pair of particles in the two-qudit maximally entangled state 
\begin{equation}
	\label{ent}
	|\Psi\rangle = \frac{1}{\sqrt{d}} \sum_{k=0}^{d-1} |kk\rangle,
\end{equation}
distribute one particle to Alice and the other to Bob, and allow Alice (Bob) to freely and independently choose and perform one measurement from $S''$ (from the set obtained by taking the complex conjugate of the elements in $S''$). Here, we apply this embedding not only to KS sets but to any SI-C set.

The second ingredient is the observation that the behavior produced by this state and these measurements violate the following Bell inequality:
\begin{widetext}
\begin{equation}
	\label{Bell}
	 \sum_{i \in V({\cal G})} w_i P(\Pi_i^A =1,\Pi_i^B =1) - \sum_{(i,j) \in E({\cal G})} \frac{\max(w_i, w_j)}{2} \left[ P(\Pi_i^A =1, \Pi_j^B =1)+P(\Pi_j^A =1, \Pi_i^B =1) \right] \stackrel{\mbox{\tiny{LHV}}}{\leq} \alpha({\cal G},w), 
\end{equation}
\end{widetext}
where $P(\Pi_i^A =1, \Pi_j^B =1)$ is the probability that Alice obtains outcome~$1$ for measurement $\Pi_i$ on her particle and Bob obtains outcome~$1$ for measurement $\Pi_j$ on his particle. LHV stands for local hidden-variable theories.

That (\ref{Bell}) is a Bell inequality follows from the fact that, for LHV theories, the maximum of the left-hand side of (\ref{Bell}) is always attained by a deterministic assignment for the outcomes of the elements of $S''$ in Alice's particle and a deterministic assignment for the outcomes of the elements of the complex conjugate of $S''$ in Bob's particle. To maximize the left-hand side of (\ref{Bell}), we need to maximize (taking into account the weights) the number of projectors $\Pi_i$ to which outcome~$1$ is assigned both in Alice's and Bob's particles, while minimizing the number of adjacent $\Pi_j$ to which outcome~$1$ is assigned, which is exactly the definition of independence number of a (weighted) graph $({\cal G}, w)$. We can translate this violation into a nonlocal game with quantum advantage following the method in \cite{BCPSW14} (Sec.~II.B4).

The interest of Bell inequality~\eqref{Bell} comes from the following observations. Noncontextuality inequalities of the form~\eqref{SI-C} are in one-to-one correspondence with Bell inequalities of the form~\eqref{Bell}. The noncontextual bound in \eqref{SI-C} is equal to the local bound in \eqref{Bell}. The quantum violation of~\eqref{SI-C} for the maximally mixed state using $S''$ is equal to the quantum violation of the Bell inequality \eqref{Bell} for the maximally entangled state (\ref{ent}) and using $S''$ in Alice's side and the complex conjugate of $S''$ in Bob's side. Moreover, if $S''$ admits a weight $w$ for which the left-hand side of (\ref{SI-C}) is represented in quantum theory by $\lambda \openone$ with $\lambda > \alpha({\cal G},w)$ (as is the case in many critical SI-C sets, e.g., \cite{YO12,CEG96,LBPC14}), then, all quantum states violate inequality (\ref{SI-C}) by the same value, and this violation coincides with that of the Bell inequality \eqref{Bell} for state (\ref{ent}).

Overall, step 3 is an interesting result by itself, as it applies to any SI-C set (and not only to sets that can be reduced to KS sets, as \cite{CHTW04,AGA12}) and preserves the gap between quantum and noncontextual theories (while previous methods \cite{CHTW04,AGA12,CAB12} do not).

%%%%%%%%%%%%%%%%%%%%%%%%%%%%%%%%%%%%%%%%%%%%%%%%%%%%%%%%%%%%%%%%%%%

{\em Converting KCBS contextuality into nonlocality.---}Here, we apply the method described above to a set of projectors \cite{CBTB13,MANCB14} leading to a violation of the KCBS inequality \cite{KCBS08}. The method works for any form of contextuality. The example has been chosen for its relevance and simplicity, as we can use a previous result \cite{CKP16} to identify $S''$.

%%%%%%%%%%%%%%%%%%%%%%%%%%%%%%%%%%%%%%%%%%%%%%%%%%%%%%%%%%%%%%%%%%%

Consider $S=\{\Pi_1,\ldots,\Pi_5\}$, where $\Pi_i = |v_i\rangle \langle v_i|$,
with
\begin{subequations}
	\label{set1}
	\begin{align}
		|v_1\rangle =& \left(1,0,0\right)^T,\\
		|v_2\rangle =& \tfrac{1}{\sqrt{2}} \left(0,1,1\right)^T,\\
		|v_3\rangle =& \tfrac{1}{\sqrt{3}} \left(1,-1,1\right)^T,\\
		|v_4\rangle =& \tfrac{1}{\sqrt{2}} \left(1,1,0\right)^T,\\
		|v_5\rangle =& \left(0,0,1\right)^T.
	\end{align}
\end{subequations}
These measurements violate the KCBS inequality \cite{KCBS08}, which can be written \cite{Cabello16} as
\begin{equation}
	\label{KCBS}
	\sum_{i \in V(G)} P(\Pi_i=1) - \sum_{(i,j) \in E(G)} P(\Pi_i=1,\Pi_j=1) \le \alpha (G),
\end{equation}
where $G$ is the graph in Fig.~\ref{fig2}(a), for which $\alpha (G)=2$. For example \cite{CBTB13,MANCB14}, the state $|\psi\rangle = \tfrac{1}{\sqrt{3}} \left(1,1,1\right)^T$ gives $2+\frac{1}{9}$, which violates inequality (\ref{KCBS}).

%%%%%%%%%%%%%%%%%%%%%%%%%%%%%%%%%%%%%%%%%%%%%%%%%%%%%%%%%%%%%%%%%%%
% Fig. 2
%%%%%%%%%%%%%%%%%%%%%%%%%%%%%%%%%%%%%%%%%%%%%%%%%%%%%%%%%%%%%%%%%%%

\begin{figure}[t]
	\includegraphics[width=8.6cm]{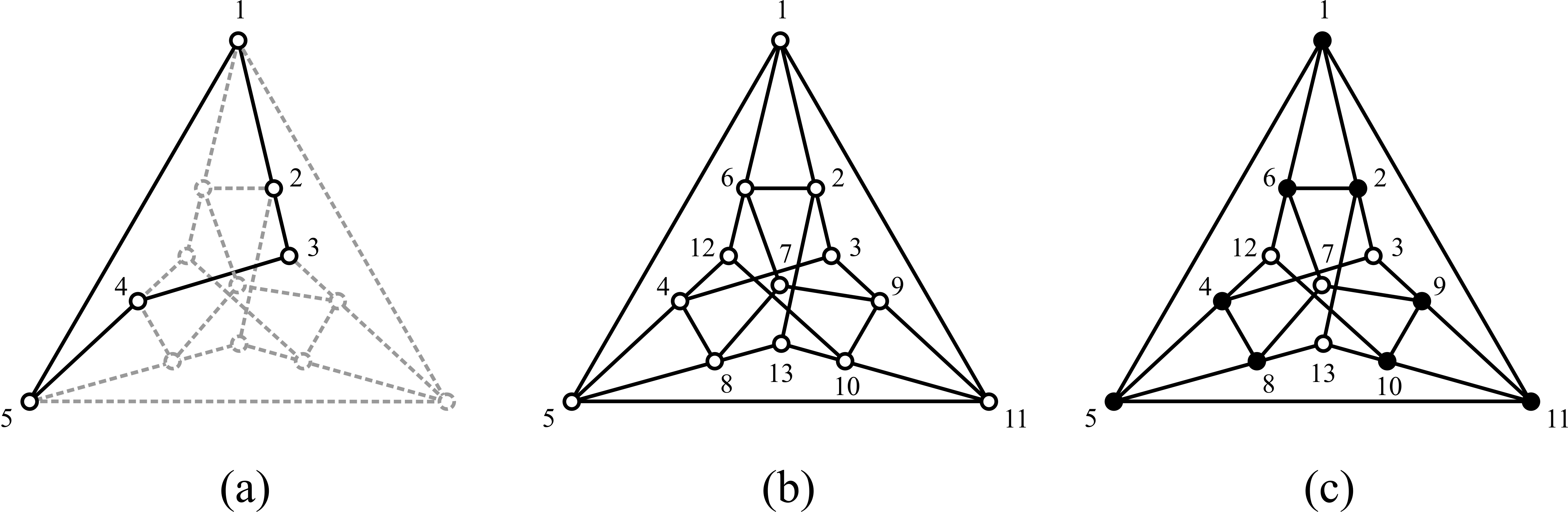}
	\caption{\label{fig2}(a) Five-vertex graph $G$ that represents the relations of orthogonality between the projectors $S=\{\Pi_1,\ldots,\Pi_5\}$ needed to violate the KCBS inequality (\ref{KCBS}). Projector $\Pi_i$ is represented by vertex $i$, mutually orthogonal projectors are represented by adjacent vertices. (b)~Extended $13$-vertex graph ${\cal G}$ representing the relations of orthogonality between elements of the smallest SI-C set $S''=\{\Pi_1,\ldots,\Pi_{13}\}$ that contains~$S$. (c)~Vertex-weighted graph $({\cal G},w)$ with the weights that produce the largest SI-C. Vertices in white have weight~$2$ and vertices in black have weight~$3$. These are the weights used in the SI-C inequality~(\ref{SI-C}) and the Bell inequality~(\ref{Bell}).}
\end{figure}

%%%%%%%%%%%%%%%%%%%%%%%%%%%%%%%%%%%%%%%%%%%%%%%%%%%%%%%%%%%%%%%%%%%

Step 1: The smallest critical SI-C set $S''$ that contains $S$ is the Yu-Oh set \cite{YO12}. This follows from the proof in \cite{CKB15} that the Yu-Oh set is the SI-C set of rank-1 projectors with minimum cardinality. Therefore, $S'=\{\Pi_6,\ldots,\Pi_{13}\}$, where $\Pi_i = |v_i\rangle \langle v_i|$,
with
\begin{subequations}
	\label{set2}
	\begin{align}
		|v_6\rangle =& \tfrac{1}{\sqrt{2}} \left(0,1,-1\right)^T,\\
		|v_7\rangle =& \tfrac{1}{\sqrt{3}} \left(1,1,1\right)^T,\\
		|v_8\rangle =& \tfrac{1}{\sqrt{2}} \left(1,-1,0\right)^T,\\
		|v_9\rangle =& \tfrac{1}{\sqrt{2}} \left(1,0,-1\right)^T,\\
		|v_{10}\rangle =& \tfrac{1}{\sqrt{2}} \left(1,0,1\right)^T,\\
		|v_{11}\rangle =& \left(0,1,0\right)^T,\\
		|v_{12}\rangle =& \tfrac{1}{\sqrt{3}} \left(-1,1,1\right)^T,\\
		|v_{13}\rangle =& \tfrac{1}{\sqrt{3}} \left(1,1,-1\right)^T.
	\end{align}
\end{subequations}
The graph ${\cal G}$ that represents the relations of orthogonality between the projectors $S''=\{\Pi_1,\ldots,\Pi_{13}\}$ is shown in Fig.~\ref{fig2}(b).

%%%%%%%%%%%%%%%%%%%%%%%%%%%%%%%%%%%%%%%%%%%%%%%%%%%%%%%%%%%%%%%%%%%

Step 2: The set of weights $\{w_1,\ldots,w_{13}\}$ leading to the
largest gap between quantum and noncontextual theories for inequality (\ref{SI-C}) for $S''$ is $w_i=2$ for $i=3,7,12,13$, and $w_i=3$, otherwise. See Fig.~\ref{fig2}(c). This follows from the observation that, in this case, the noncontextuality inequality (\ref{SI-C}) has $\alpha ({\cal G},w)=11$, while it is violated by any quantum state of dimension $d=3$, since, for any initial state (including the maximally mixed state), the left-hand side of (\ref{SI-C}) is $\frac{1}{3} (2 \times 4 + 3 \times 9) = 11 + \frac{2}{3}$.

%%%%%%%%%%%%%%%%%%%%%%%%%%%%%%%%%%%%%%%%%%%%%%%%%%%%%%%%%%%%%%%%%%%

Step 3: Distributing pairs of particles in the maximally entangled state (\ref{ent}), with $d=3$, between Alice and Bob and allowing each of them to perform a randomly chosen spacelike separated measurement from $S''$ (in this case, $S''$ and its complex conjugate are equal), we obtain a nonlocal behavior as the local bound of the Bell inequality (\ref{Bell}) is $\alpha ({\cal G},w)=11$, while the value for the left-hand side of (\ref{Bell}) is, again, $11 + \frac{2}{3}$. 

%%%%%%%%%%%%%%%%%%%%%%%%%%%%%%%%%%%%%%%%%%%%%%%%%%%%%%%%%%%%%%%%%%%

{\em Explanation, virtues, and limitations.---}Here, we give some intuition of how the method works. The set of states (in dimension $d \ge 3$) that yield contextual behaviors grows as the set of measurements grows from $S$ to $S''$. For example, while the state $|\psi'\rangle= \frac{1}{\sqrt{3}}(1,-1,1)^T$ does not violate inequality~(\ref{KCBS}), it violates a similar noncontextuality inequality replacing $S$ by $\{\Pi_1,\ldots,\Pi_9\}$ \cite{KK12}.
% [state $|\psi \rangle$ in (\ref{ini}) violates both inequalities].
When all the measurements in $S''$ are used, then even the maximally mixed state produces contextuality and weights can be adjusted \cite{CKB15} to produce equal state-independent violation of a noncontextuality inequality for all states \cite{Cabello08,BBCP09,YO12,KBLGC12}.

The Bell inequality (\ref{Bell}) follows from the SI-C inequality (\ref{SI-C}) 
%(or state-independent noncontextuality inequality \cite{CKB15}) 
by noticing that (\ref{SI-C}) can be tested in experiments consisting of two sequential measurements on a maximally mixed state. We can assume that each of these measurements is performed by a different party. Sometimes Alice is the first to measure and Bob the second, and sometimes vice versa. Sometimes both parties measure the same $\Pi_i$, sometimes they measure different but compatible projectors. This view leads to the Bell inequality (\ref{Bell}) which shares the classical bound and it is also violated by the same amount when preparing pairs in state (\ref{ent}) and giving one particle to Alice and the other to Bob, as in this case Alice's and Bob's outcomes are perfectly correlated, and Alice's and Bob's local states are maximally mixed states. 

%%%%%%%%%%%%%%%%%%%%%%%%%%%%%%%%%%%%%%%%%%%%%%%%%%%%%%%%%%%%%%%%%%%

Virtues: (I) While inequalities (\ref{SI-C}) and (\ref{KCBS}) are noncontextuality inequalities that might only be testable by performing sequential nondemolition measurements on single systems \cite{KZG09,ZUZ13,LMZNCAH18}, inequality (\ref{Bell}) is a Bell inequality that can be tested by performing local measurements on spatially separated systems and can be converted into a nonlocal game.
(II) The ``compatibility'' or ``sharpness'' loophole \cite{GKCLKZGR10} in contextuality experiments with sequential measurements vanishes in the Bell test, as, there, measurements do not need to be ideal (or sharp) \cite{Cabello19} and observables on different particles are automatically compatible.
(III) The gap between quantum and local theories for the Bell inequality (\ref{Bell}) is the same as the gap between quantum and noncontextual theories for the SI-C~inequality~(\ref{SI-C}), and both are produced using the same measurements.
(IV) The violation of the Bell inequality (\ref{Bell}) by the measurements in $S''$ and state (\ref{ent}) vanishes whenever we remove from $S''$ any element of $S$. This follows from the fact that, in that case, inequality (\ref{SI-C}) is not violated by the maximally mixed state. Therefore, the maximally entangled state (\ref{ent}) fails to violate the Bell inequality (\ref{Bell}), as the local states of Alice and Bob are maximally mixed. This property follows from the fact that $S''$ is a critical SI-C set.
(V) There is no ``contextuality-nonlocality tradeoff'' \cite{KCK14,ZZLZSX16}. The quantum violations of the SI-C inequality (\ref{SI-C}) and the Bell inequality (\ref{Bell}) can be tested simultaneously in the same experiment. According to quantum theory, the experiment would give (equal) violations of both inequalities. The violation of (\ref{SI-C}) can be observed by allowing one of the parties, e.g., Alice, to perform sequential measurements. The violation of (\ref{Bell}) can be observed by considering the first (or second) measurements of Alice and the (only) measurements of Bob. It would be interesting to observe these simultaneous violations in an actual experiment.

%%%%%%%%%%%%%%%%%%%%%%%%%%%%%%%%%%%%%%%%%%%%%%%%%%%%%%%%%%%%%%%%%%%

Limitations: Except for the case $S=S''$, the nonlocal behavior resulting from the application of this method does not have the same gap between quantum and local theories than the gap between quantum and noncontextual theories of the original SD-C behavior. Arguably, no method exists that preserves this gap for all forms of state-dependent contextuality.

%%%%%%%%%%%%%%%%%%%%%%%%%%%%%%%%%%%%%%%%%%%%%%%%%%%%%%%%%%%%%%%%%%%

\begin{acknowledgments}
We thank Costantino Budroni and Zhen-Peng Xu for comments on an earlier version. This work was supported by Project Qdisc (Project No.\ US-15097), with FEDER funds, MINECO Project No.\ FIS2017-89609-P, with FEDER funds, and QuantERA Grant SECRET, by MINECO (Project No.\ PCI2019-111885-2).
\end{acknowledgments}

% European Regional Development Fund, FundRef ID http://dx.doi.org/10.13039/501100008530 (Kingdom of Belgium/BE)

% Ministerio de Economía y Competitividad, FundRef ID http://dx.doi.org/10.13039/501100003329 (ES/Kingdom of Spain)

%%%%%%%%%%%%%%%%%%%%%%%%%%%%%%%%%%%%%%%%%%%%%%%%%%%%%%%%%%%%%%%%%%%

\appendix

\section{Basic definitions and results in contextuality for ideal measurements}
\label{sec1}

%%%%%%%%%%%%%%%%%%%%%%%%%%%%%%%%%%%%%%%%%%%%%%%%%%%%%%%%%%%%%%%%%%%

Here, by contextuality we mean contextuality produced by ideal measurements. For further details and other notions of contextuality the reader is referred to \cite{BCGKL21}. 

\begin{Definition}
	An {\em ideal measurement} of an observable $A$ is a measurement of $A$ that yields the same outcome when repeated and does not disturb any compatible observable.
\end{Definition}

\begin{Definition}
	Two observables $A$ and $B$ are {\em compatible} if there exists an observable $C$ such that, for every initial state $\rho$, for every outcome $a$ of $A$,
	\begin{equation}
		P\left(A=a\middle|\rho\right)=\sum_{c_a}{P\left(C=c_a\middle|\rho\right)}, 
	\end{equation} 
	and, for every outcome $b$ of $B$,
	\begin{equation} 
		P\left(B=b\middle|\rho\right)=\sum_{c_b}{P\left(C=c_b\middle|\rho\right)},
	\end{equation}
	where $P\left(A=a\middle|\rho\right)$ is the probability of, given state $\rho$, obtaining outcome $a$ for $A$.
	$C$ is called a {\em refinement} of $A$ (and $B$).
	Therefore, two observables are compatible when they have a common refinement.
\end{Definition}

Hereafter, for simplicity, every time all probabilities refer to the same initial state, we will write $P\left(A=a\right)$ rather than $P\left(A=a\middle|\rho\right)$.

\begin{remark}
	In quantum theory measurements are represented by POVMs.
\end{remark}

\begin{remark}
	\label{r2}
	One of the fundamental predictions of quantum theory is that the outcome-statistics of any measurement can be obtained with an ideal measurement. This follows from two results. (a) The observation \cite{Kleinmann14} that (L\"uders') state transformation produced by a self-adjoint operator \cite{Luders51} corresponds to the only process that can be associated with a measurement of an observable $A$ that does not disturb a subsequent measurement of any refinement of $A$. (b) Neumark's theorem stating that every POVM can be dilated to a self-adjoint operator \cite{Neumark40,Holevo80,AG93}. 
\end{remark}

Remark~\ref{r2} motivates the interest in studying correlations produced by ideal measurements of compatible observables for more general probabilistic theories.

\begin{remark}
	In quantum theory, every ideal measurement is represented by the spectral projectors of a self-adjoint operator and compatible observables correspond to commuting self-adjoint operators.
\end{remark}

\begin{Definition}
	A {\em context} is a set of ideal measurements of compatible observables.
\end{Definition}

\begin{Definition}
	A {\em contextuality scenario} is a number of ideal measurements, the cardinal of their respective sets of possible outcomes, and their relations of compatibility. 
\end{Definition}

\begin{Definition}
	A {\em behavior} for a contextuality scenario is a set of (normalized) probability distributions produced by ideal measurements satisfying the relations of compatibility of the scenario, one for each of the contexts, and such that the probability for every outcome of every measurement does not depend on the context (nondisturbance condition).
\end{Definition}

%\begin{Example} A behavior for the scenario with contexts $\{1,2\}$, $\{2,3\}$, $\{3,4\}$, and $\{4,1\}$, and where every measurement has possible outcomes $0$ and $1$, can be given by the following list: 
%\begin{table}[h!]
%\begin{center}
%\begin{tabular}{rccl}
%$\{P(00|12)$, & $P(01|12)$, & $P(10|12)$, & $P(11|12)$, \\
%$P(00|23)$, & $P(01|23)$, & $P(10|23)$, & $P(11|23)$, \\
%$P(00|34)$, & $P(01|34)$, & $P(10|34)$, & $P(11|34)$, \\
%$P(00|41)$, & $P(01|41)$, & $P(10|41)$, & $P(11|41)\}$,
%\end{tabular}
%\end{center}
%\end{table}

%\noindent where, e.g., $P(01|23)$ is the probability of obtaining outcomes $0$ and $1$ for measurements $2$ and $3$.
%\end{Example}

%\begin{remark} 
%In quantum theory, every behavior for a contextuality scenario is produced by a set of projectors $\{\Pi_i\}$ (not necessarily of rank-one) acting on an initial state $\rho$.
%\end{remark}

\begin{Definition}
	A behavior for a contextuality scenario is {\em contextual} if the probability distributions for each context cannot be obtained as the marginals of a global probability distribution on all observables. Otherwise the behavior is {\em noncontextual}.
\end{Definition}

\begin{Definition}
	The relations of compatibility between $N$~observables can be represented by an $N$-vertex graph, called the {\em graph of compatibility} of the scenario, in which each vertex represents an observable and adjacent vertices correspond to compatible observables.
\end{Definition}

\begin{theorem} \cite{Vorob'yev62}
	The only contextuality scenarios that admit contextual behaviors are those in which the graph of compatibility is nonchordal. That is, it contains, as induced subgraph, at least one cycle of four or more vertices (i.e., a square, a pentagon, an hexagon, etc.).
\end{theorem}

\begin{Definition}
	An {\em induced subgraph} of a graph $G$ is a graph formed from a subset of the vertices of $G$ and all of the edges connecting pairs of vertices in that subset.
\end{Definition}

\begin{colloraly}
	The simplest scenario in which contextuality with ideal measurements is possible is the one consisting of four dichotomic observables whose graph of compatibility is a square. The {\em Clauser-Horne-Shimony-Holt \cite{CHSH69} Bell scenario} has this graph of orthogonality.
\end{colloraly}

\begin{theorem} \cite{XC19}
	Quantum theory produces contextual behaviors with ideal measurements in every scenario allowed by Vorob'ev's theorem.
\end{theorem}

\begin{theorem} \label{bks} \cite{Bell66,KS67}
	In quantum theory, contextuality for ideal measurements requires quantum systems of dimension three or higher. 
\end{theorem}

\begin{theorem} \cite{KCBS08}
	The simplest scenario in which contextuality with ideal measurements on qutrits is possible is the one consisting of five dichotomic observables whose graph of compatibility is a pentagon. This is the {\em Klyachko-Binicio\u{g}lu-Can-Shumovsky (KCBS) scenario} \cite{KCBS08}.
\end{theorem}

%%%%%%%%%%%%%%%%%%%%%%%%%%%%%%%%%%%%%%%%%%%%%%%%%%%%%%%%%%%%%%%%%%%

\section{Results used in the proof of Theorem~\ref{tm}}
\label{sec2}

%%%%%%%%%%%%%%%%%%%%%%%%%%%%%%%%%%%%%%%%%%%%%%%%%%%%%%%%%%%%%%%%%%%

\begin{Definition} 
	A noncontextuality (NC) inequality is an inequality satisfied by any noncontextual behavior.
\end{Definition}

\begin{Definition}
	The {\em graph of orthogonality} of a set of projectors $\{\Pi_i\}$ is the one in which each vertex represents a projector and adjacent vertices correspond to orthogonal projectors.
\end{Definition}

\begin{Definition}
	The {\em independence number} of a graph $G$ is the cardinality of the maximal independent set of $G$. That is, the maximum number of $1$'s that can be assigned to the vertices of $G$ without violating the condition that two adjacent vertices cannot both have assigned the value~$1$.
\end{Definition}

\begin{theorem} \label{cswc} \cite{CSW14,Cabello16} 
	Every quantum contextual behavior produced by a set of projectors $S = \{\Pi_i,\ldots,\Pi_n\}$ violates a noncontextuality (NC) inequality of the form
	\begin{widetext}
		\begin{equation}
			\label{csw}
			\sum_{i \in V(G)} w_i P(\Pi_i =1) - \sum_{(i,j) \in E(G)} \max{(w_i, w_j)} P(\Pi_i =1, \Pi_j =1) \stackrel{\mbox{\tiny{NCHV}}}{\leq} \alpha(G,w),
		\end{equation}
	\end{widetext}
	where $G$ is the graph of orthogonality of $S$, 
	$V(G)$ is the set of vertices of $G$, 
	$E(G)$ is the set of edges of $G$, $w = (w_1, \ldots, w_n)$ are nonnegative numbers, and $\alpha(G,w)$ is the independence number of the weighted graph $(G,w)$ (i.e., of the graph $G$ in which vertex $i$ has associated weight $w_i$).
\end{theorem}

\begin{theorem} \label{cswc2} \cite{CSW14}
	A necessary condition for inequality \eqref{csw} to be violated is that $G$ contains, as induced subgraph, at least an odd cycle of size $5$ or larger (i.e., pentagon, heptagon, etc.) or one of their complements.
\end{theorem}

\begin{remark}
	The KCBS inequality \cite{KCBS08} can be expressed as \eqref{csw}, where $G$ is a pentagon.
\end{remark}

\begin{Definition}
	The {\em complement of a graph} $G$, denoted $\overline{G}$, is a graph on the same vertices such that two distinct vertices of $\overline{G}$ are adjacent if and only if they are not adjacent in $G$. 
\end{Definition}

\begin{Definition}
	A {\em state-dependent contextuality (SD-C) set} is a set of projectors that produces contextual behaviors for some initial states.
\end{Definition}

\begin{Definition}
	A {\em state-independent contextuality (SI-C) set} is a set of projectors that produces contextual behaviors for any initial state.
\end{Definition}

\begin{Definition}
	A {\em $k$-coloring of a graph} $G$ is an assignment of one out of $k$ colors to each vertex of $G$ such that adjacent vertices are assigned different colors. 
\end{Definition}

\begin{Definition}
	The {\em chromatic number} of a graph $G$, denoted $\chi(G)$, is the minimal number $k$ such that a $k$-coloring of $G$ is possible. The chromatic number is also the minimal number of partitions of the graph into independent sets.
\end{Definition}

\begin{Definition}
	The {\em fractional chromatic number}, denoted $\chi_f (G)$, is the minimum of $\frac{a}{b}$ such that vertices have $b$ associated colors, out of $a$ colors, where adjacent vertices have associated disjoint sets of colors.
\end{Definition} 

\begin{remark}
	$\chi_f (G) \le \chi(G)$.
\end{remark}

\begin{theorem} \label{t6} \cite{Cabello11}
	For a set of rank-one projectors $\{\Pi_i\}$ in dimension $d$, the condition $\chi (G) > d$, for the orthogonality graph $G$ of $\{\Pi_i\}$, is necessary for SI-C.
\end{theorem}

\begin{theorem} \label{t6b} \cite{RH14,CKB15}
	For a set of rank-one projectors $\{\Pi_i\}$ in dimension $d$, the conditions $\chi_f (G) > d$ for the orthogonality graph $G$ of $\{\Pi_i\}$, is necessary for SI-C.
\end{theorem}

\begin{theorem} \label{t66} \cite{CKB15}
	A set of rank-one projectors $S = \{\Pi_i,\ldots,\Pi_n\}$ is a SI-C set if and only if there are nonnegative numbers $w = (w_1, \ldots, w_n)$ and a number $0 \leq y < 1$ such that $\sum_{j \in \cal I} w_j \le y$ for all ${\cal I} $, where ${\cal I}$ is any independent set of the graph of orthogonality
	of $S$, and $\sum_{i} w_i \Pi_i \ge \openone$.
\end{theorem}

In particular, $w$ gives rise to a NC inequality violated by any quantum state, inequality that can be written as (\ref{csw}), where $G$ is the graph of orthogonality of $S$.

\begin{Definition}
	\label{kss}
	A {\em Kochen-Specker (KS) assignment} to a set of rank-one projectors is an assignment of $0$ or $1$ satisfying that: (I)~two orthogonal projectors cannot both have assigned~$1$, (II)~for every set of mutually orthogonal projectors summing the identity, one of them must be assigned~$1$.
\end{Definition}

\begin{Definition}
	A {\em KS set} is a set of rank-one projectors which does not admit a KS assignment.
\end{Definition}

\begin{theorem} \cite{KS67}
	KS sets exist in dimension $d \ge 3$, but not in $d=2$.
\end{theorem} 

\begin{theorem} \cite{CEG96,XCG20}
	The KS set of minimal cardinality has $18$ rank-one projectors in $d=4$.
\end{theorem} 

\begin{remark} \cite{Peres93,CK06,CK09}
	The KS set of minimal cardinality in $d=3$ known has $31$ rank-one projectors. 
\end{remark}

\begin{theorem} \cite{PMMM05}
	The KS set of minimal cardinality in $d=3$ must have $22$ or more rank-one projectors. 
\end{theorem}

\begin{theorem}\label{t12} \cite{BBCP09}
	Every KS set is a SI-C set.
\end{theorem} 

\begin{theorem} \cite{YO12,BBC12,XCS15}
	Not every SI-C set of rank-one projectors is a KS set. 
\end{theorem}

\begin{proof}
	The SI-C sets in \cite{YO12,BBC12,XCS15} are not KS sets.
\end{proof}

\begin{Definition} \cite{ZP93}
	A KS set $S$ is {\em critical} if by removing any element of $S$ the resulting set is not a KS set.
\end{Definition}

\begin{remark} \cite{Pavicic17}
	The original KS set \cite{KS67} is critical.
\end{remark}

\begin{Definition}
	A SI-C set $S$ is {\em critical} if by removing any element of $S$ the resulting set is not a SI-C set.
\end{Definition}

\begin{remark} \cite{YO12}
	\label{t10f}
	Not every critical KS set is a critical SI-C set.
\end{remark}

\begin{proof}
	The critical KS sets of Conway and Kochen \cite{Peres93,CK06,CK09} (with $31$~rank-one projectors), Sch\"utte and Bub \cite{Bub96,Bub97} (with $33$~rank-one projectors), and Peres \cite{Peres91} (with $33$~rank-one projectors and isomorphic \cite{GA10} to the one proposed by Penrose \cite{Penrose00}), all contain the $13$~rank-one projectors of the set of Yu and Oh \cite{YO12}, which is a critical SI-C set.
\end{proof}

\begin{remark} \cite{X21}
	The critical KS set with $18$ projectors in $d=4$ \cite{CEG96} is a critical SI-C set.
\end{remark}

\begin{theorem} \cite{CKP16}
	The SI-C set of rank-one projectors of minimal cardinality is the set of $13$ rank-one projectors in $d=3$ of Yu and Oh \cite{YO12}.
\end{theorem}

\begin{theorem} \cite{Toh13a,Toh13b}
	There are SI-C sets of projectors not all of them of rank one.
\end{theorem}

\begin{theorem} \cite{XYK21}
	For $n=1,2,3$, the SI-C set made of rank-$n$ projectors (all of them) having minimal cardinality has $13$ rank-one projectors in $d=3$.
\end{theorem}

\begin{remark}
	To every SI-C set of rank-$n$ projectors one can associate a SI-C set of rank-one projectors.
\end{remark}

\begin{remark} \cite{Peres91,KP95}
	Every proof of state-independent contextuality made of self-adjoint operators which are not rank-$n$ projectors (e.g., \cite{Peres90,Mermin90}) has associated a SI-C set of rank-one projectors. 
\end{remark}

\begin{Definition}
	A set of projectors $S=\{\Pi_A=|\psi_A\rangle \langle \psi_A| \} \cup \{\Pi_i,\ldots,\Pi_n\} \cup \{\Pi_B =|\psi_B\rangle \langle \psi_B|\}$ in dimension $d \ge 3$, where $\Pi_A$ and $\Pi_B$ are nonorthogonal and $\{\Pi_i,\ldots,\Pi_n\}$ are not necessarily rank-one projectors, is a {\em true-implies-false set (TIFS)} if, for any KS assignment~$f$, $f(\Pi_A) + f(\Pi_B) \le 1$. Therefore, $f(\Pi_A)=1$ implies $f(\Pi_B) = 0$, and $f(\Pi_B)=1$ implies $f(\Pi_A) = 0$.
\end{Definition}

The name TIFS is taken from \cite{CPSS18}. TIFSs made of rank-one projectors are called {\em definite prediction sets} in \cite{CG96}, {\em $01$-gadgets} in \cite{RRHPHH20}, and {\em Hardy-like proofs} in \cite{XCG20} (after the observation \cite{CEG96} that Hardy's proof of Bell nonlocality \cite{Hardy93} can be represented by a TIFS). 

\begin{theorem} \cite{CPSS18}
	The TIFS with the smallest number of rank-one projectors is the TIFS in dimension $3$ introduced in \cite{KS65} and shown in Figs.~\ref{specker} and \ref{kochen}.
\end{theorem}

%%%%%%%%%%%%%%%%%%%%%%%%%%%%%%%%%%%%%%%%%%%%%%%%%%%%%%%%%%%%%%%%%%%

\begin{figure}[t!]\centering
	\includegraphics[width=.5\textwidth]{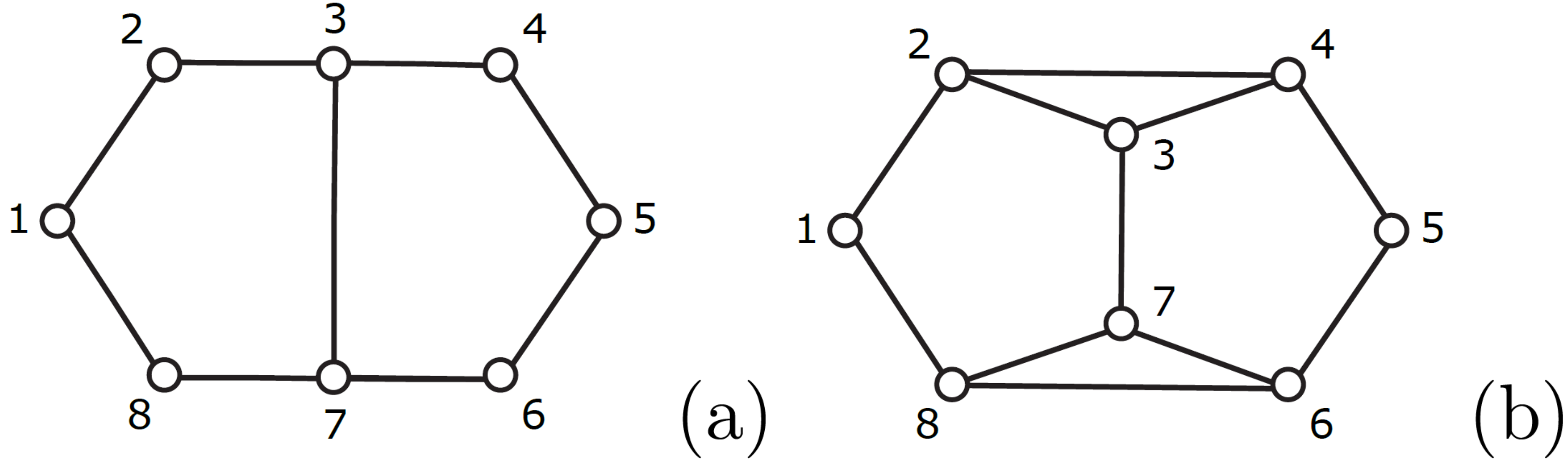}
	\caption{Two different ways to present the graph of orthogonality of the simplest TIFS in dimension~$3$. Projectors are represented by nodes. (a) Projectors in the same straight line are mutually orthogonal. (b) Projectors in the same clique (subset of vertices such that every two vertices are adjacent) are mutually orthogonal. Assuming $d=3$ and all projectors of rank-one, if $1$ is true, then, by condition~(I) in Definition~\ref{kss}, both $2$ and $8$ must be false. Then, by condition~(II) in Definition~\ref{kss}, at least one of $4$ and $6$ must be true. Then, by (II), $5$ be false. In this case $A=1$ and $B=5$. The figure is taken from \cite{BCGKL21}. \label{specker}}
\end{figure}

%%%%%%%%%%%%%%%%%%%%%%%%%%%%%%%%%%%%%%%%%%%%%%%%%%%%%%%%%%%%%%%%%%%

\begin{figure}[t!]\centering
	\includegraphics[width=.25\textwidth]{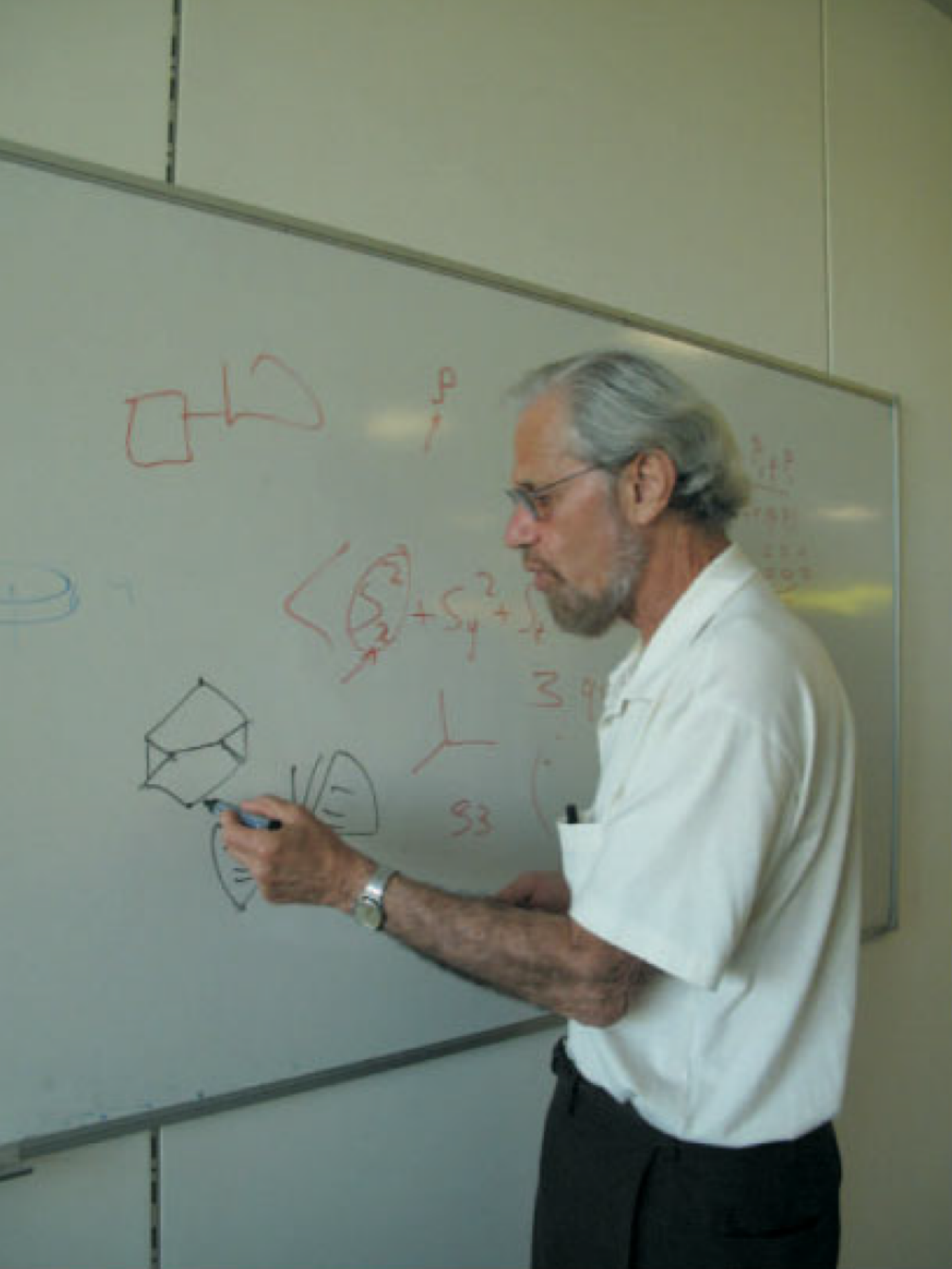}
	\caption{Simon Kochen drawing the simplest TIFS (see Fig.~\ref{specker}) in Z\"urich in 2009. Ernst Specker called this TIFS ``the bug'' \cite{CPSS18}. The photo is taken from the initial online version of \cite{CT11}. \label{kochen}}
\end{figure}

%%%%%%%%%%%%%%%%%%%%%%%%%%%%%%%%%%%%%%%%%%%%%%%%%%%%%%%%%%%%%%%%%%%

\begin{theorem} \cite{KS67}
	Given a TIFS in $d=3$, one can construct a critical KS set in dimension $d=3$ by concatenating this TIFS (i.e., by adding to this TIFS suitably rotated versions of the TIFS). See Fig.~\ref{ks117}.
\end{theorem}

%%%%%%%%%%%%%%%%%%%%%%%%%%%%%%%%%%%%%%%%%%%%%%%%%%%%%%%%%%%%%%%%%%%

\begin{figure}[t!]\centering
	\includegraphics[width=.5\textwidth]{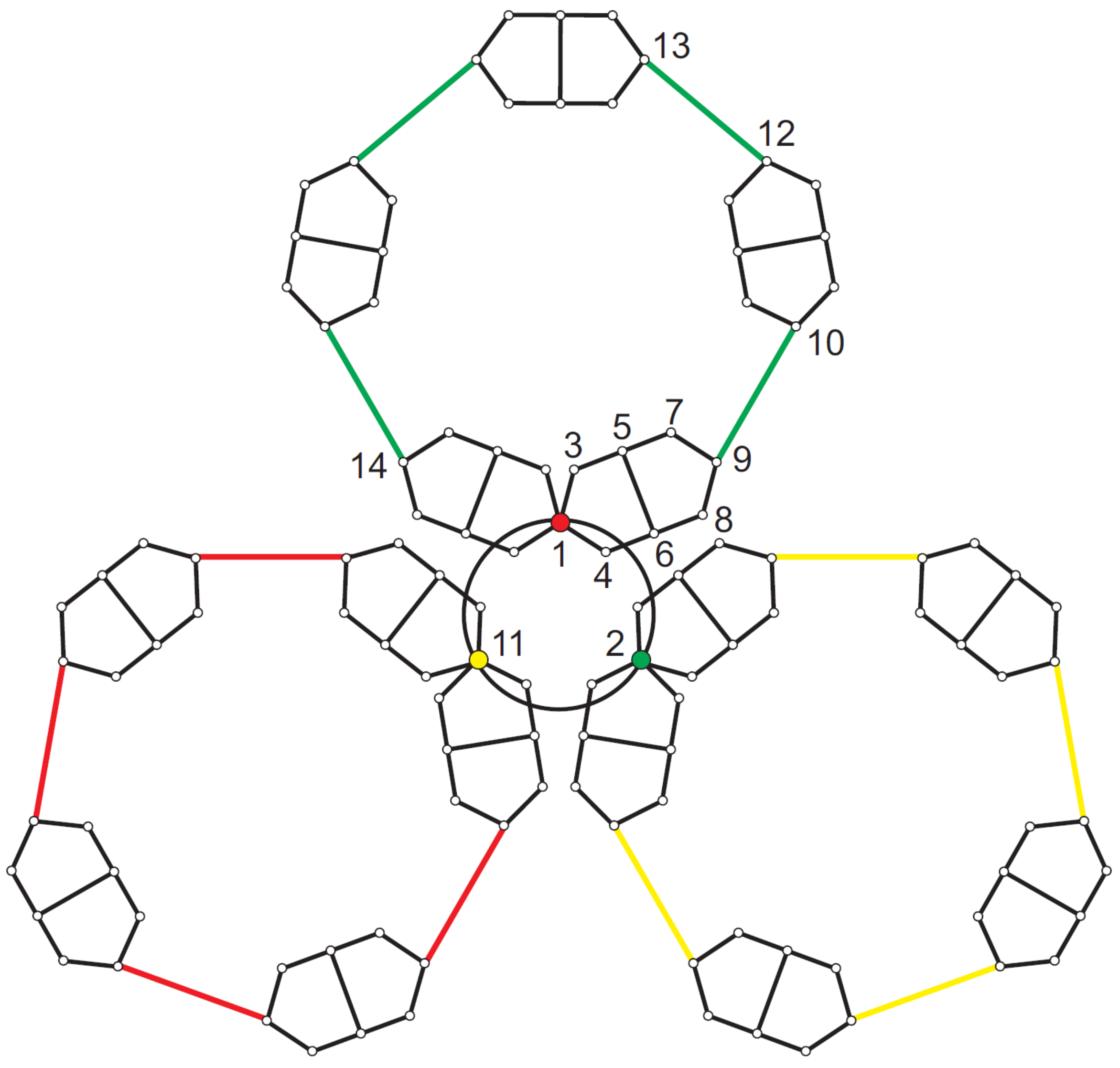}
	\caption{Graph of orthogonality of the $117$~rank-one projectors in dimension $3$ in the proof of KS \cite{KS67}. Nodes in the same straight line or circumference represent mutually orthogonal projectors. The red node is orthogonal to all nodes
		connected to the red edge. Similarly for the green and yellow
		nodes. That the set does not admit a KS assignment can be seen as follows.
		One of the nodes $1$, $2$, and $11$ has to be true. The symmetry of
		the graph allows us to assume without loss of generality that
		it is node $1$. Then, node $9$ must be false, because of the TIFS between node $1$ and node $9$ (see Fig.~\ref{specker}). Then, since
		$2$, $9$, and $10$ are mutually orthogonal and $2$ is connected to $1$,
		node $10$ must be true. Applying the same argument, $12$ must
		be false and $13$ must be true, since $\{12, 13,
		2\}$ form a basis.
		Repeating it again twice, $14$ must be true. However, $1$ and $14$
		cannot be both true. The figure is taken from \cite{BCGKL21}. \label{ks117}}
\end{figure}

%%%%%%%%%%%%%%%%%%%%%%%%%%%%%%%%%%%%%%%%%%%%%%%%%%%%%%%%%%%%%%%%%%%

\begin{remark} 
	\label{rm}
	The original KS set \cite{KS67} is a critical KS set which, for some choices of the TIFS involved in its construction, is also a critical SI-C set. This follows from the freedom for choosing the specific rank-one projectors corresponding to the $15$~TIFS in Fig.~\ref{ks117}.
\end{remark}

\begin{theorem}\label{t10k} \cite{CG96}
	For any $d \ge 3$, there is a TIFS with $d+5$ rank-one projectors.
\end{theorem}

\begin{theorem}\label{t10m} \cite{CPSS18}
	For any $d \ge 3$, the TIFS of rank-one projectors with minimum cardinality has $d+5$ rank-one projectors.
\end{theorem}

\begin{theorem}\label{t10n} \cite{CBTB13}
	Any SD-C set of rank-one projectors whose orthogonality graph is an odd cycle of size~$5$ or larger can be extended to a TIFS. These SD-C sets are of fundamental importance for contextuality because of Theorem~\ref{cswc2}.
\end{theorem}

\begin{theorem}\label{t10} \cite{CG96}
	Given a TIFS in $d \ge 3$, one can construct a critical KS set in any dimension $D \ge d$.
\end{theorem}

\begin{theorem}\label{t10b} \cite{CG96,RRHPHH20}
	Given any two nonorthogonal rank-one projectors in any dimension, there is a TIFS between them.
\end{theorem}

\begin{theorem}\label{t10d} \cite{RRHPHH20}
	Every set $\{B_0,B_1,B_2,B_3\}$ of four disjoint bases of rank-one projectors in dimension $3$ such that there is a TIFS between every projector $\Pi_i \in B_0$, with $i=1,\ldots,d$, and every projector in $B_i$, is a critical KS set. See Fig.~\ref{gadgetKS3}.
\end{theorem}

%%%%%%%%%%%%%%%%%%%%%%%%%%%%%%%%%%%%%%%%%%%%%%%%%%%%%%%%%%%%%%%%%%%

\begin{figure}[t!]\centering
	\includegraphics[width=.36\textwidth]{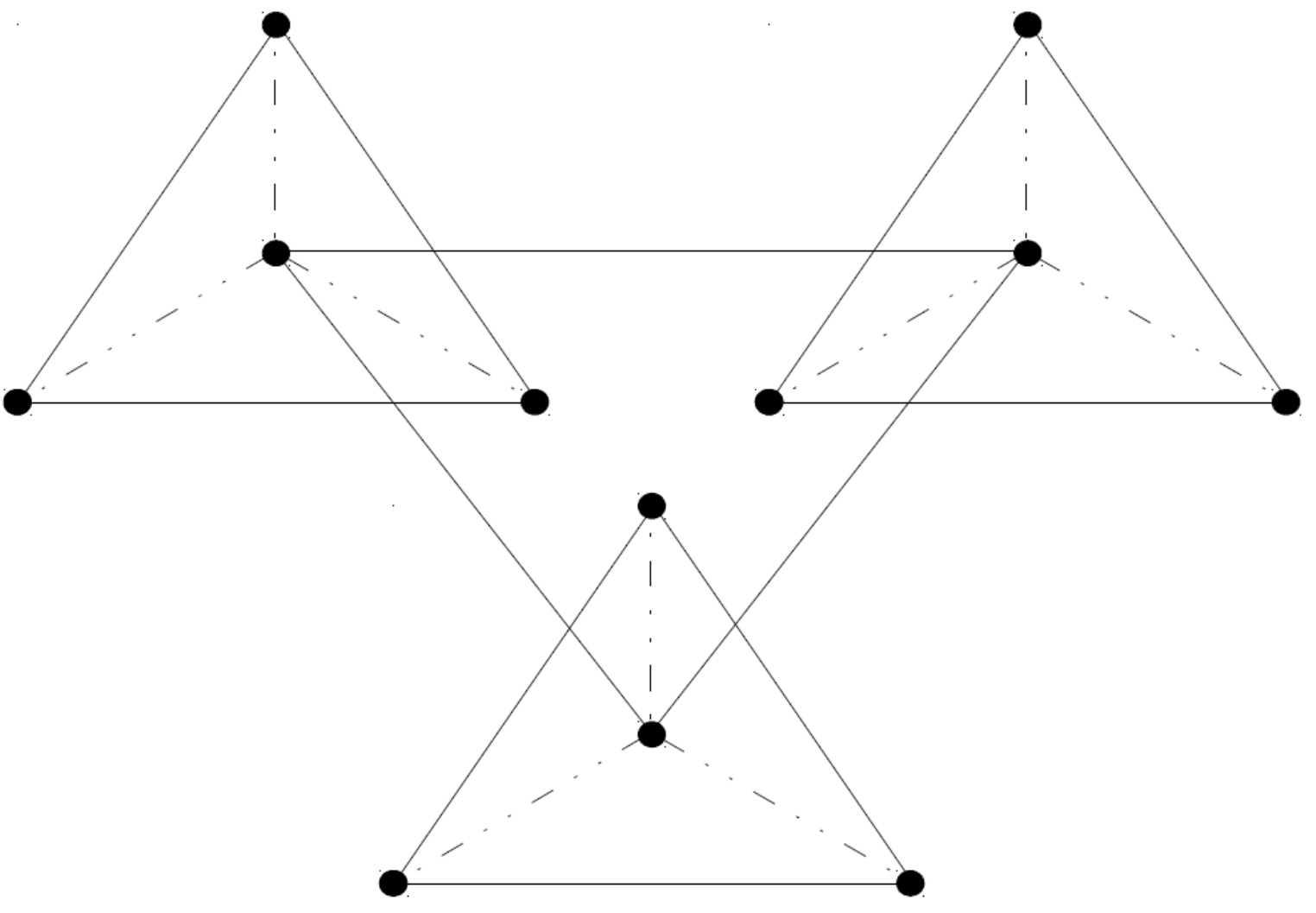}
	\caption{Nodes represent rank-one projectors. A continuous line between two nodes indicates that they are orthogonal. A dashed line between two nodes indicates that there is a TIFS between them. The figure illustrates Theorem~\ref{t10d}. Due to condition~(II) in Definition~\ref{kss}, one of the nodes of the central triangle has to be true. Then, due to condition~(I) in Definition~\ref{kss}, the three nodes of the corresponding small triangle have to be false, in contradiction to condition~(II). The figure is taken from \cite{RRHPHH20}. \label{gadgetKS3}}
\end{figure}

%%%%%%%%%%%%%%%%%%%%%%%%%%%%%%%%%%%%%%%%%%%%%%%%%%%%%%%%%%%%%%%%%%%

\begin{theorem}\label{t10d2}
	In any dimension $d \ge 3$, every set of $N$ disjoint bases of rank-one projectors can be extended to a critical KS set.
\end{theorem}

\begin{proof}
	If $N = d+1$, one can use the construction in Theorem~\ref{t10d} extended to arbitrary $d \ge 3$, as shown in Figs.~\ref{fig1}(a) and \ref{fig1}(d).
	If $N < d+1$, one can add disjoint bases until $N'= d+1$ and then apply the method in Figs.~\ref{fig1}(a) and \ref{fig1}(d).
	If $N>d+1$, one can use the construction in Figs.~\ref{fig1}(b), \ref{fig1}(c), and \ref{fig1}(e).
\end{proof}

%%%%%%%%%%%%%%%%%%%%%%%%%%%%%%%%%%%%%%%%%%%%%%%%%%%%%%%%%%%%%%%%%%%

\section{Any SD-C set can be extended to a critical SI-C set}
\label{sec3}

%%%%%%%%%%%%%%%%%%%%%%%%%%%%%%%%%%%%%%%%%%%%%%%%%%%%%%%%%%%%%%%%%%%

\begin{theorem}
	\label{tm}
	Any SD-C set $S$ of projectors can be extended to a critical SI-C set containing all the projectors of $S$.
\end{theorem}

\begin{proof}
	Let $G$ be the graph of orthogonality of $S$. Let $N$ be the minimum number of disjoint bases that cover all the vertices of~$G$. See Fig.~\ref{covers}. 
	%By Theorem~\ref{bks}, $d \ge 3$. 
	By Theorem~\ref{cswc2}, $N \ge 3$. By Theorem~\ref{t10d2}, every set of $N$ disjoint bases of rank-one projectors can be extended to a critical KS set. By Theorem~\ref{t12}, this KS set is a SI-C set. By Remark~\ref{t10f}, this does not mean that the SI-C set is critical. However, as in Remark~\ref{rm}, criticality of the SI-C set can be enforced by suitably choosing the TIFSs used in the proof of Theorem~\ref{t10d2}.
\end{proof}

%%%%%%%%%%%%%%%%%%%%%%%%%%%%%%%%%%%%%%%%%%%%%%%%%%%%%%%%%%%%%%%%%%%

\begin{figure}[t!]\centering
	\includegraphics[width=.40\textwidth]{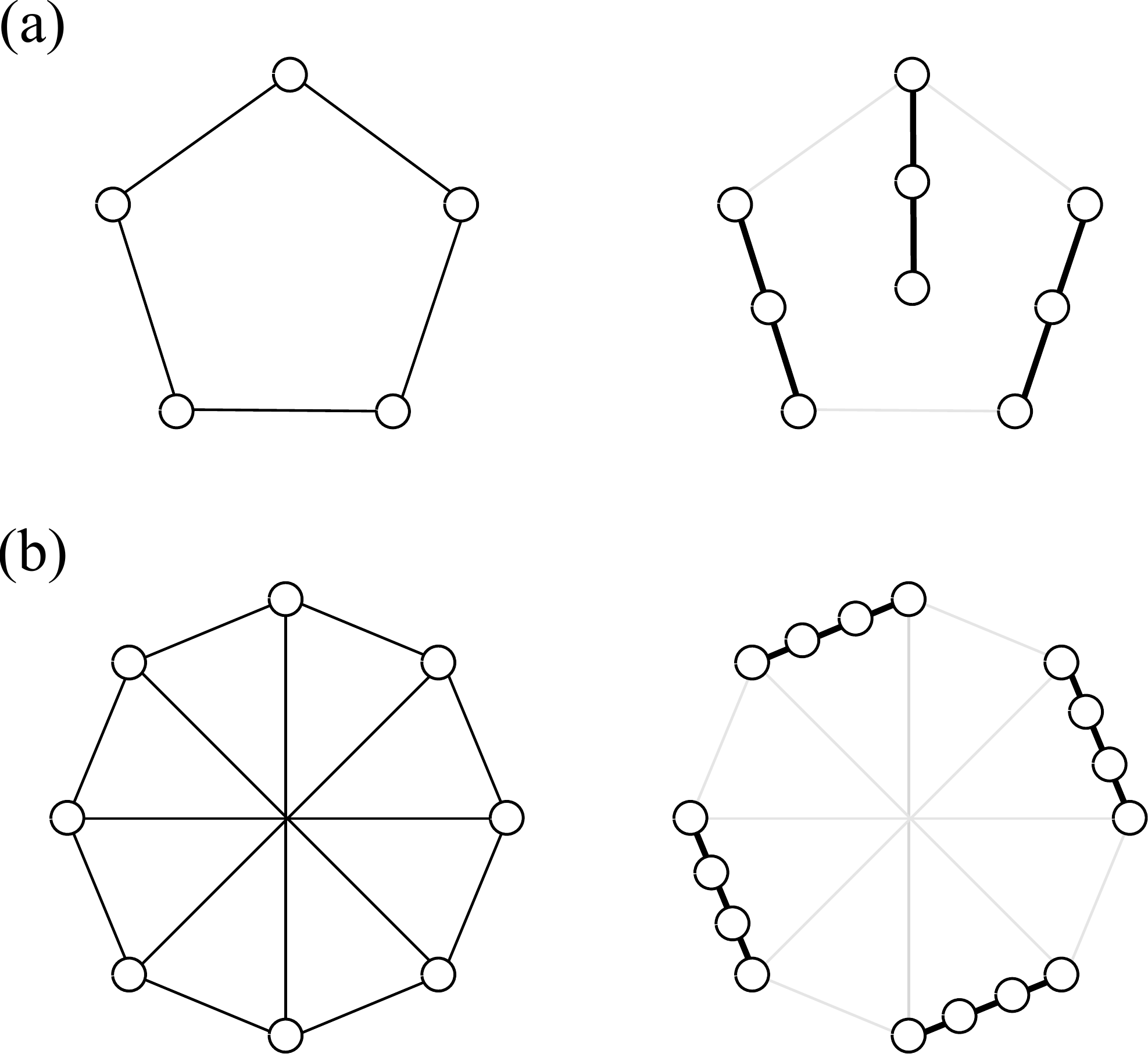}
	\caption{Nodes in the same straight line represent mutually orthogonal rank-one projectors. (a)~A pentagon in $d=3$ and a minimum number of disjoint bases covering all its vertices. (b)~A circulant graph in $d=4$ and a minimum number of disjoint bases covering all its vertices. \label{covers}}
\end{figure}

%%%%%%%%%%%%%%%%%%%%%%%%%%%%%%%%%%%%%%%%%%%%%%%%%%%%%%%%%%%%%%%%%%%

%%%%%%%%%%%%%%%%%%%%%%%%%%%%%%%%%%%%%%%%%%%%%%%%%%%%%%%%%%%%%%%%%%%

\section{How to identify minimal critical SI-C sets containing a given SD-C set}
\label{sec4}

%%%%%%%%%%%%%%%%%%%%%%%%%%%%%%%%%%%%%%%%%%%%%%%%%%%%%%%%%%%%%%%%%%%

\begin{Definition}
	A minimal critical SI-C set is a critical SI-C set of minimum cardinality.
\end{Definition}

Theorem~\ref{tm} guarantees that critical SI-C sets containing any given SD-C set exist. However, the method used in the proof of Theorem~\ref{tm} does not guarantee that the resulting critical SI-C set is minimal. Here we show how to use Theorems~\ref{t6} and~\ref{t6b} to identify candidates to be minimal critical SI-C sets containing any given SD-C set, and then Theorem~\ref{t66} to check whether or not they are SI-C sets.

%%%%%%%%%%%%%%%%%%%%%%%%%%%%%%%%%%%%%%%%%%%%%%%%%%%%%%%%%%%%%%%%%%%

\begin{figure}[t!]\centering
	\includegraphics[width=.32\textwidth]{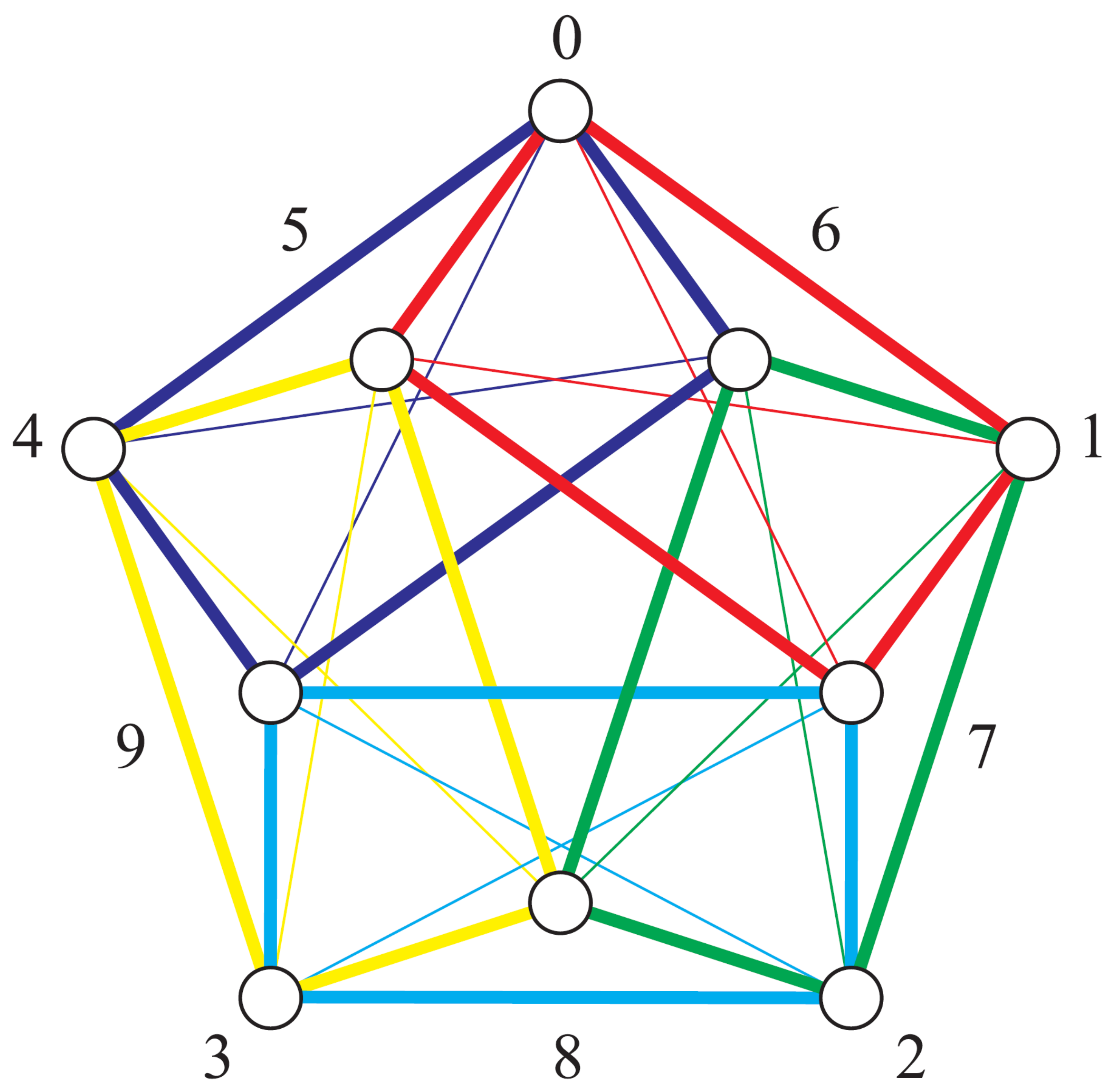}
	\caption{Graph of orthogonality of the $10$~projectors in Eq.~(\ref{twinvectors}). Nodes represent projectors and edges orthogonalities. Edges of the same color belong to the same maximal context. This graph is called {\em Johnson graph $J(5,2)$}, after Selmer M.\ Johnson. It is the complement graph of the {\em Petersen graph}. The figure is taken from \cite{Cabello13b}. \label{twin_PRA}}
\end{figure}

%%%%%%%%%%%%%%%%%%%%%%%%%%%%%%%%%%%%%%%%%%%%%%%%%%%%%%%%%%%%%%%%%%%

\begin{Example}
	Consider the SD-C set of projectors $S=\{\Pi_j=|v_j\rangle \langle v_j|\}_{j=0}^9$, where:
	\begin{subequations}
		\label{twinvectors}
		\begin{align}
			|v_0\rangle &= 2^{-1} (1,0,a^2,a,0,1)^T,\\
			|v_1\rangle &= 2^{-1} (0,0,1,1,1,1)^T,\\
			|v_2\rangle &= (1,0,0,0,0,0)^T,\\
			|v_3\rangle &= (0,0,0,0,0,1)^T,\\
			|v_4\rangle &= 2^{-1} (1,1,1,1,0,0)^T, \\
			|v_5\rangle &= 2^{-1} (1,0,a,a^2,1,0)^T, \\
			|v_6\rangle &= 2^{-1} (0,1,a,a^2,0,1)^T, \\
			|v_7\rangle &= (0,0,0,0,1,0)^T, \\
			|v_8\rangle &= 2^{-1} (0,1,a^2,a,1,0)^T, \\
			|v_9\rangle &= (0,1,0,0,0,0)^T,
		\end{align}
	\end{subequations} 
	with $a=e^{2 \pi i/3}$ (thus $a^2=e^{-2 \pi i/3}$). Its graph of orthogonality is shown in Fig.~\ref{twin_PRA} and is called $J(5,2)$. The projectors defined by Eq.~(\ref{twinvectors}) lead to the maximum quantum state-dependent violation of the {\em twin inequality} \cite{Cabello13b}, which can be written in the form of Eq.~(\ref{csw}). Specifically,
	\begin{widetext}
		\begin{equation}
			\label{twin}
			\sum_{i\in V(J(5,2))} P(\Pi_i =1) - \sum_{(i,j) \in E(J(5,2))} P(\Pi_i =1, \Pi_j =1) \stackrel{\mbox{\tiny{NCHV}}}{\leq} \alpha(J(5,2)),
		\end{equation}
	\end{widetext} 
	where $J(5,2)$ is the graph in Fig.~\ref{twin_PRA}. The maximum quantum violation is obtained by preparing the eigenstate corresponding to the maximum eigenvalue of the contextuality operator defined in the left-hand side of (\ref{twin}). 
	
As discussed in \cite{Cabello13b}, the twin inequality is of fundamental importance for understanding the limits of quantum correlations, as it provides the simplest example of fully contextual quantum correlations associated to a vertex transitive graph of orthogonality \cite{ADLPBC12}. 
The graphs of orthogonality of fully contextual correlations satisfy $\alpha(G) < \vartheta(G) = \alpha^\ast(G)$, where $\vartheta(G)$ is the Lov\'asz number of $G$ (and is equal to the maximum quantum violation \cite{CSW14}) and $\alpha^\ast(G)$ is the fractional packing number of $G$ \cite{ADLPBC12}.
	
For the graph in Fig.~\ref{twin_PRA}, 
	\begin{equation} 
		\alpha(J(5,2))=2,\;\;\;\;\vartheta(J(5,2)) = \alpha^\ast (J(5,2))=\frac{5}{2}.
	\end{equation}
	
What is the minimal critical SI-C set $S'$ that contains $S$? Theorem~\ref{t6} tells us that a necessary condition is that the graph of orthogonality of $S'$, denoted $G'$, must have $\chi(G')>6$.
	
%%%%%%%%%%%%%%%%%%%%%%%%%%%%%%%%%%%%%%%%%%%%%%%%%%%%%%%%%%%%%%%%%%%
	
\begin{figure}[t!]\centering
		\includegraphics[width=.5\textwidth]{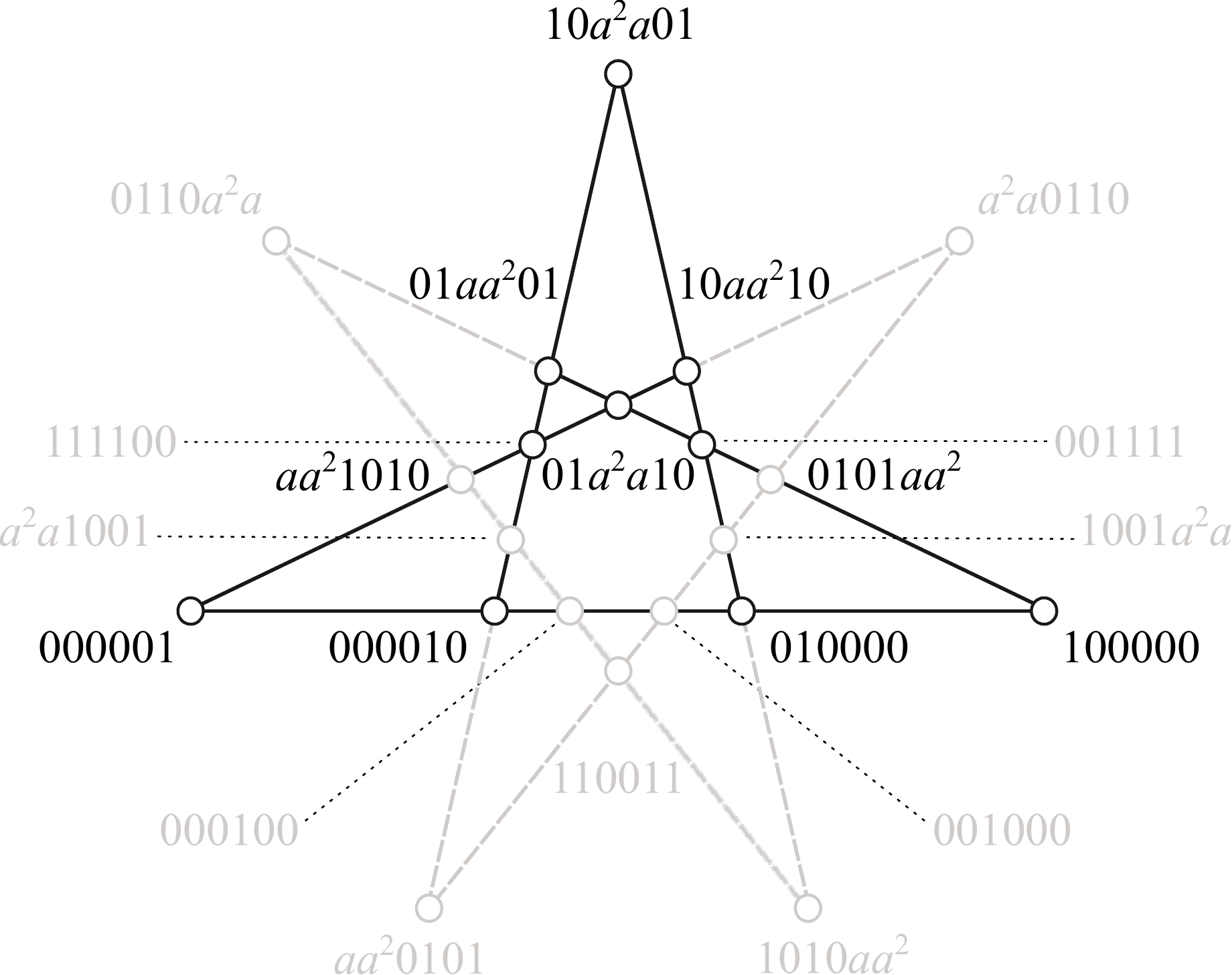}
		\caption{Graph of orthogonality of a set of $21$~projectors which includes the $10$~projectors in Eq.~(\ref{twinvectors}). $10a^2a01$ represents the vector $2^{-1}(1,0,a^2,a,0,1)$, with $a=e^{2 \pi i/3}$.
		Nodes in the same straight line represent mutually orthogonal projectors. This graph is called {\em Johnson graph $J(7,2)$}. \label{ks21-twin-slide2}}
\end{figure}
	
%%%%%%%%%%%%%%%%%%%%%%%%%%%%%%%%%%%%%%%%%%%%%%%%%%%%%%%%%%%%%%%%%%%
	
\begin{figure}[t!]\centering
		\includegraphics[width=.37\textwidth]{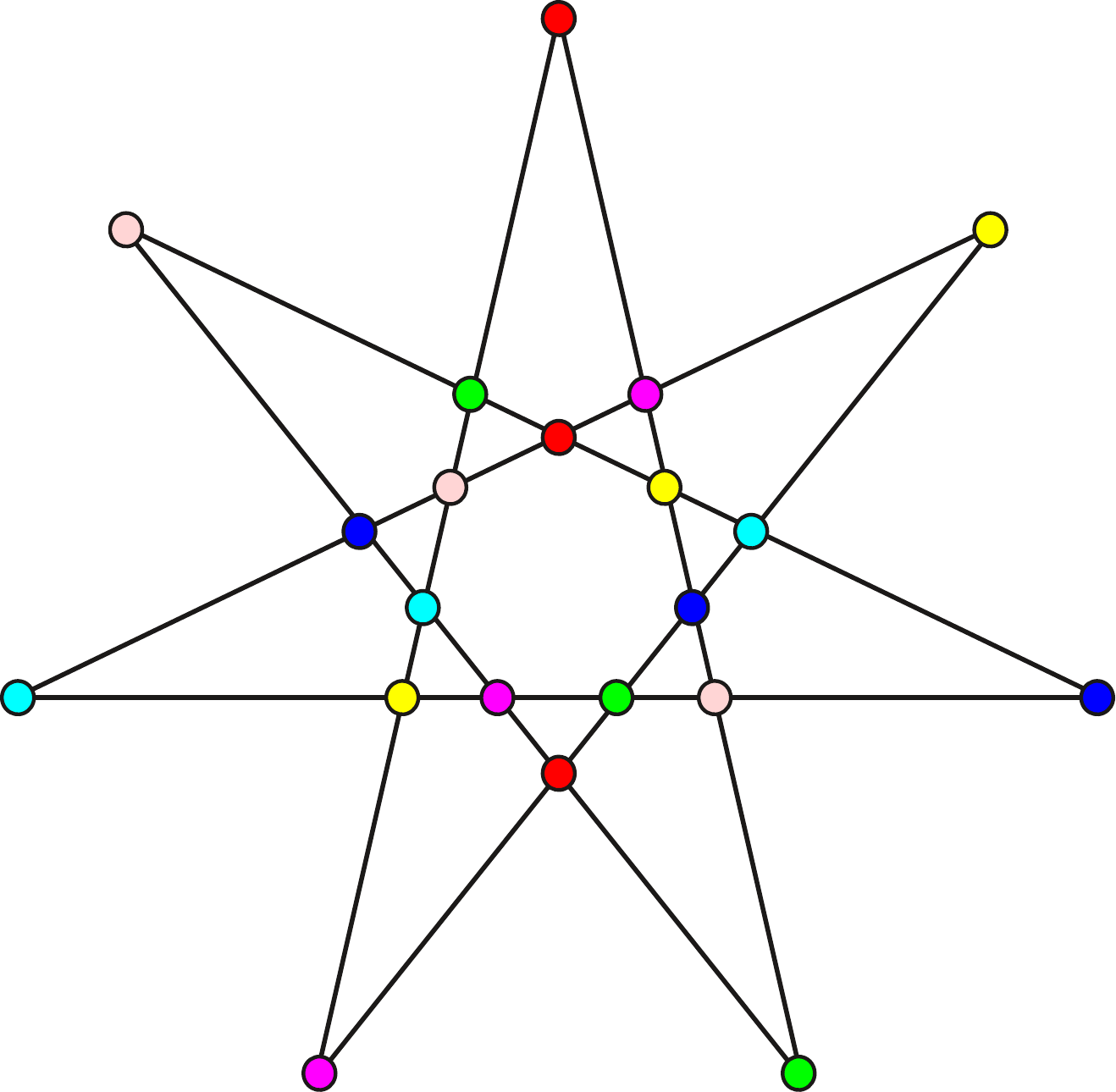}
		\caption{A $7$-coloring of $J(7,2)$. There are no $6$-coloring of $J(7,2)$. \label{ks21-twin-chomatic}}
\end{figure}
	
%%%%%%%%%%%%%%%%%%%%%%%%%%%%%%%%%%%%%%%%%%%%%%%%%%%%%%%%%%%%%%%%%%%

The graph in Fig.~\ref{ks21-twin-slide2}, known as $J(7,2)$, has $\chi(J(7,2))=7$ (see Fig.~\ref{ks21-twin-chomatic}) and contains $21$ induced copies of $J(5,2)$. Moreover, as shown in Fig.~\ref{ks21-twin-slide2}, there is a set $S'$ which contains $S$ and has $J(7,2)$ as graph of orthogonality.
Moreover, $S'$ is the critical KS set in dimension $6$ with the minimum number of contexts of maximal size such that its graph of orthogonality is vertex transitive and every rank-one projector is in exactly an even number of contexts \cite{LBPC14}.
	
$S'$ is a critical SI-C set. First, notice that if we remove any three vertices from $J(7,2)$, then the chromatic number of the resulting graph is never larger than~$6$, the dimension of the quantum system. That is, $\chi (J(7,2)-3) \not> 6$. Therefore, by Theorem~\ref{t6}, any $J(7,2)-3$ cannot a SI-C set. However, if we remove any two vertices from $J(7,2)$, then the resulting $19$-vertex graphs $J(7,2)-2$ have fractional chromatic number $\frac{19}{6}$, which is larger than $6$. That is, $\chi_f (J(7,2)-2) > 6$. Therefore, Theorem~\ref{t6b} identifies $J(7,2)-2$ as potential SI-C sets. However, Theorem~\ref{t66} tells us that they are not a SIC sets, at least if we take the $19$~observables with weight $1$ in the noncontextuality inequality, as then, although $\alpha(J(7,2)-2)=3 < \frac{19}{6}$, the smallest eigenvalues of the corresponding contextuality witnesses are always smaller than $\alpha(J(7,2)-2)$. Checking whether $J(7,2)-2$ and $J(7,2)-1$ are SI-C sets for arbitrary weights is a problem of convex optimization. This optimization confirms that $J(7,2)$ is a critical SI-C set \cite{X21}.
While, proving minimality requires checking all graphs with $20$ or less vertices that contain induced $J(5,2)$, the fact that the critical SI-C set we have identified only contains $S$ and rotations of it suggests that it is the most elegant extension of $S$ into a critical SI-C set.
\end{Example}

%%%%%%%%%%%%%%%%%%%%%%%%%%%%%%%%%%%%%%%%%%%%%%%%%%%%%%%%%%%%%%%%%%%

\section{On the connection between the noncontextuality inequality and the Bell inequality}
\label{sec5}

%%%%%%%%%%%%%%%%%%%%%%%%%%%%%%%%%%%%%%%%%%%%%%%%%%%%%%%%%%%%%%%%%%%

As explained in the main text, given a SI-C set $S$, there is a bipartite Bell inequality which is violated when one of the parties can choose between the observables represented by the projectors in $S$ and the other party can choose between the observables represented by the complex conjugate of the projectors in $S$. This Bell inequality has the form
\begin{widetext}
	\begin{equation}
		\label{bbi}
		\sum_{i \in V(G)} w_i P(\Pi^A_i=1, \Pi^B_i=1) - \sum_{(i,j) \in E(G)} \frac{\max{(w_i, w_j)}}{2} \left[P(\Pi^A_i=1, \Pi^B_j=1) + P(\Pi^A_j=1, \Pi^B_i=1) \right]\le \alpha (G,w).
	\end{equation}
\end{widetext}
Remarkably, there is a fundamental connection between this Bell inequality and the noncontextuality inequality violated by any quantum state \eqref{csw}.

\begin{theorem}
	For any SI-C set $S$, the quantum violation of inequality~\eqref{csw} for the maximally mixed state using $S$ is equal to the quantum violation of the Bell inequality \eqref{bbi} for the maximally entangled state using $S$ in Alice's side and the complex conjugate of $S$ in Bob's side.
\end{theorem}

\begin{Definition}
	An {\em egalitarian} SI-C set is one for which there is a $w$ in which all weights are nonzero for which the left-hand side of (\ref{csw}) is represented in quantum theory by an operator $\lambda \openone$ such that $\lambda > \alpha(G,w)$. In this case, all quantum states violate inequality (\ref{csw}) by the same value.
\end{Definition}

\begin{theorem} \cite{X21}
	Not all critical SI-C sets are egalitarian.
\end{theorem}

\begin{theorem}
	For any egalitarian SI-C set $S$, if $w$ is the set of weights with all weights nonzero used in inequalities~\eqref{csw} and \eqref{bbi}, then the quantum violation of inequality~\eqref{csw} for any state using $S$ is the same as the quantum violation of the Bell inequality \eqref{bbi} for the maximally entangled state using $S$ in Alice's side and the complex conjugate of $S$ in Bob's side.
\end{theorem}

\begin{Example}
	The set of Yu and Oh \cite{YO12} is an egalitarian critical SI-C set. The weights that lead to the largest gap between the quantum violation and the classical bound is $w_i=3$ for all projectors, except for those onto $(1,1,1)$, $(1,1,-1)$, $(1,-1,1)$, and $(-1,1,1)$, where $w_j=2$. With these weights, $\alpha(G,w)=11$ both in \eqref{csw} and \eqref{bbi} and the violation of inequality \eqref{csw} for any state, and of inequality~\eqref{bbi} for the maximally entangled state are both $11+\frac{2}{3}$.
\end{Example}

\begin{Example}
	The KS set with $18$ projectors \cite{CEG96} is an egalitarian critical SI-C set. The weights that lead to the largest gap between the quantum violation and the classical bound are $\{w_i=1\}_{i=1}^{18}$. With these weights, $\alpha(G,w)=4$ both in \eqref{csw} and \eqref{bbi}, and the violation of inequality \eqref{csw} for any state, and of inequality~\eqref{bbi} for the maximally entangled state are both $4+\frac{1}{2}$.
\end{Example}

\begin{Example}
	The KS set with $21$ projectors \cite{LBPC14} is an egalitarian critical SI-C set. The weights that lead to the largest gap between the quantum violation and the classical bound are $\{w_i=1\}_{i=1}^{21}$. With these weights, $\alpha(G,w)=3$ both in \eqref{csw} and \eqref{bbi}, and the violation of inequality \eqref{csw} for any state, and of inequality~\eqref{bbi} for the maximally entangled state are both $3+\frac{1}{2}$.
\end{Example}

%%%%%%%%%%%%%%%%%%%%%%%%%%%%%%%%%%%%%%%%%%%%%%%%%%%%%%%%%%%%%%%%%%%


\begin{thebibliography}{99}
	
%%%%%%%%%%%%%%%%%%%%%%%%%%%%%%%%%%%%%%%%%%%%%%%%%%%%%%%%%%%%%%%%%%%

\bibitem{CHTW04}
R. Cleve, P. H{\o}yer, B. Toner, and J. Watrous,
Consequences and limits of nonlocal strategies,
in \href{https://doi.org/10.1109/CCC.2004.1313847}{{\em Proceedings. Nineteenth Annual IEEE Conference on Computational Complexity. 21--24 June, 2004, Amherst, Massachusetts} (IEEE Computer Society Press, Los Alamitos, CA, 2004), p.~236.}
	
\bibitem{BCPSW14}
N. Brunner, D. Cavalcanti, S. Pironio, V. Scarani, and S. Wehner,
Bell nonlocality,
\href{https://doi.org/10.1103/RevModPhys.86.419}{Rev. Mod. Phys. \textbf{86}, 419 (2014).}

\bibitem{CL89}
A. Condon and R. J. Lipton,
On the complexity of space bounded interactive proofs,
in \href{https://doi.org/10.1109/SFCS.1989.63519}{{\em 30th Annual Symposium on Foundations of Computer Science} (IEEE Computer Society Press, Los Alamitos, CA, 1989), p.~462.}

\bibitem{AB09}
J. Anders and D. E. Browne,
Computational Power of Correlations,
\href{https://doi.org/10.1103/PhysRevLett.102.050502}{Phys. Rev. Lett. \textbf{102}, 050502 (2009).}

\bibitem{Raussendorf13}
R. Raussendorf,
Contextuality in measurement-based quantum computation,
\href{https://doi.org/10.1103/PhysRevA.88.022322}{Phys. Rev. A \textbf{88}, 022322 (2013).}

\bibitem{HWVE14}
M. Howard, J. Wallman, V. Veitch, and J. Emerson,
Contextuality supplies the `magic' for quantum computation,
\href{https://doi.org/10.1038/nature13460}{Nature (London) \textbf{510}, 351 (2014).}

\bibitem{Stairs83}
A. Stairs, 
Quantum logic, realism, and value definiteness,
\href{https://doi.org/10.1086/289140}{Philos. Sci. {\bf 50}, 578 (1983).}

\bibitem{HR83}
P. Heywood and M. L. G. Redhead, 
Nonlocality and the Kochen-Specker paradox,
\href{https://doi.org/10.1007/BF00729511}{Found. Phys. {\bf 13}, 481 (1983).}

\bibitem{KS67}
S. Kochen and E. P. Specker,
The problem of hidden variables in quantum mechanics,
\href{http://www.iumj.indiana.edu/IUMJ/fulltext.php?year=1968&volume=17&artid=17004}
{J. Math. Mech. \textbf{17}, 59 (1967).}

\bibitem{BS90}
H. R. Brown and G. Svetlichny,
Nonlocality and Gleason's lemma. Part I. Deterministic theories,
\href{https://doi.org/10.1007/BF01883492}{Found. Phys. {\bf 20}, 1379 (1990).}

\bibitem{CK06}
J. H. Conway and S. Kochen, 
The free will theorem,
\href{https://doi.org/10.1007/s10701-006-9068-6}{Found. Phys. \textbf{36}, 1441 (2006).}

\bibitem{Cabello01}
A. Cabello,
``All versus Nothing'' Inseparability for Two Observers,
\href{https://doi.org/10.1103/PhysRevLett.87.010403}{Phys. Rev. Lett. \textbf{87}, 010403 (2001).}

\bibitem{GMS07}
N. Gisin, A. A. M\'ethot, and V. Scarani,
Pseudo-telepathy: Input cardinality and Bell-type inequalities,
\href{https://doi.org/10.1142/S021974990700289X}{Int. J. Quant. Inf. \textbf{05}, 525 (2007).}

\bibitem{AGA12}
L. Aolita, R. Gallego, A. Ac\'{\i}n, A. Chiuri, G. Vallone, P. Mataloni, and A. Cabello,
Fully nonlocal quantum correlations,
\href{https://doi.org/10.1103/PhysRevA.85.032107}{Phys. Rev. A \textbf{85}, 032107 (2012).}

\bibitem{Aravind02}
P. K. Aravind, 
Bell's theorem without inequalities and only two distant observers,
\href{https://doi.org/10.1023/A:1021272729475}{Found. Phys. Lett. \textbf{15}, 397 (2002).}

\bibitem{GBT05}
G. Brassard, A. Broadbent, and A. Tapp, 
Quantum pseudo-telepathy,
\href{https://doi.org/10.1007/s10701-005-7353-4}{Found. Phys. \textbf{35}, 1877 (2005).}

\bibitem{CMNSW07}
P. J. Cameron,
A. Montanaro,
M. W. Newman,
S, Severini, and
A. Winter,
On the quantum chromatic number of a graph,
\href{https://doi.org/10.37236/999}{Electron. J. Comb. \textbf{14}, R81 (2007).}

\bibitem{CM14}
R. Cleve and R. Mittal,
Characterization of binary constraint system games,
in {\em Automata, Languages, and Programming. ICALP 2014},
edited by J. Esparza, P. Fraigniaud, T. Husfeldt, and E. Koutsoupias, 
\href{https://doi.org/10.1007/978-3-662-43948-7_27}{Lecture Notes in Computer Science Vol.\ 8572 (Springer, Berlin, 2014), p.~320.}

\bibitem{MR16}
L. Man\v{c}inska and D. E. Roberson,
Quantum homomorphisms,
\href{https://doi.org/10.1016/j.jctb.2015.12.009}{J. Comb. Theory Ser. B \textbf{118}, 228 (2016).}

\bibitem{ASDZ17}
S. Abramsky, R. Soares Barbosa, N. de Silva, and O. Zapata,
The quantum monad on relational structures,
in {\em 42nd International Symposium on Mathematical Foundations of Computer Science (MFCS 2017)},
edited by K. G. Larsen, H. L. Bodlaender, and J.-F. Raskin,
International Proceedings in Informatics Vol.\ 83
(Schloss Dagstuhl--Leibniz-Zentrum f\"ur Informatik, Saarbr\"ucken, 2017), p.~35:1.
%https://doi.org/10.4230/LIPIcs.MFCS.2017.35

\bibitem{Peres90}
A. Peres,
Incompatible results of quantum measurements,
\href{https://doi.org/10.1016/0375-9601(90)90172-K}{Phys. Lett. A \textbf{151}, 107 (1990).}

\bibitem{Mermin90}
N. D. Mermin,
Simple Unified Form for the Major No-Hidden-Variables Theorems,
\href{https://doi.org/10.1103/PhysRevLett.65.3373}{Phys. Rev. Lett. \textbf{65}, 3373 (1990).}

\bibitem{Peres91}
A. Peres,
Two simple proofs of the Kochen-Specker theorem,
\href{https://doi.org/10.1088/0305-4470/24/4/003}{J. Phys. A: Math. Gen. \textbf{24}, L175 (1991).}

\bibitem{KP95}
M. Kernaghan and Peres,
Kochen-Specker theorem for eight-dimensional space,
\href{https://doi.org/10.1016/0375-9601(95)00012-R}{Phys. Lett. A \textbf{198}, 1 (1995).}

\bibitem{YO12}
S. Yu and C. H. Oh,
State-Independent Proof of Kochen-Specker Theorem with 13 Rays,
\href{https://doi.org/10.1103/PhysRevLett.108.030402}{Phys. Rev. Lett. \textbf{108}, 030402 (2012).}

\bibitem{CAB12}
A. Cabello, E. Amselem, K. Blanchfield, M. Bourennane, and I. Bengtsson, 
Proposed experiments of qutrit state-independent contextuality and two-qutrit contextuality-based nonlocality,
\href{https://doi.org/10.1103/PhysRevA.85.032108}{Phys. Rev. A \textbf{85}, 032108 (2012).}

\bibitem{KCBS08}
A. A. Klyachko, M. A. Can, S. Binicio\u{g}lu, and A. S. Shumovsky,
Simple Test for Hidden Variables in Spin-1 Systems,
\href{https://doi.org/10.1103/PhysRevLett.101.020403}{Phys. Rev. Lett. \textbf{101}, 020403 (2008).}

\bibitem{AQB13}
M. Ara\'ujo, M. T. Quintino, C. Budroni, M. Terra Cunha, and A. Cabello,
All noncontextuality inequalities for the $n$-cycle scenario,
\href{https://doi.org/10.1103/PhysRevA.88.022118}{Phys. Rev. A \textbf{88}, 022118 (2013).}

\bibitem{CSW14}
A. Cabello, S. Severini, and A. Winter,
Graph-Theoretic Approach to Quantum Correlations,
\href{https://doi.org/10.1103/PhysRevLett.112.040401}{Phys. Rev. Lett. \textbf{112}, 040401 (2014).}

\bibitem{Cabello13}
A. Cabello,
Simple Explanation of the Quantum Violation of a Fundamental Inequality,
\href{https://doi.org/10.1103/PhysRevLett.110.060402}{Phys. Rev. Lett. \textbf{110}, 060402 (2013).}

\bibitem{Cabello19}
A. Cabello,
Quantum correlations from simple assumptions,
\href{https://doi.org/10.1103/PhysRevA.100.032120}{Phys. Rev. A \textbf{100}, 032120 (2019).}

\bibitem{XYK21}
Z.-P. Xu, X.-D. Yu, and M. Kleinmann,
State-independent quantum contextuality with projectors of nonunit rank,
\href{https://doi.org/10.1088/1367-2630/abe6e3}{New J. Phys. \textbf{23}, 043025 (2021).}

\bibitem{Cabello08}
A. Cabello,
Experimentally Testable State-Independent Quantum Contextuality,
\href{https://doi.org/10.1103/PhysRevLett.101.210401}{Phys. Rev. Lett. \textbf{101}, 210401 (2008).}

\bibitem{BBCP09}
P. Badzi\c{a}g, I. Bengtsson, A. Cabello, and I. Pitowsky,
Universality of State-Independent Violation of Correlation Inequalities for Noncontextual Theories,
\href{https://doi.org/10.1103/PhysRevLett.103.050401}{Phys. Rev. Lett. \textbf{103}, 050401 (2009).}

\bibitem{KBLGC12}
M. Kleinmann, C. Budroni, J.-\AA. Larsson, O. G{\"u}hne, and A. Cabello,
Optimal Inequalities for State-Independent Contextuality,
\href{https://doi.org/10.1103/PhysRevLett.109.250402}{Phys. Rev. Lett. \textbf{109}, 250402 (2012).}

\bibitem{ZP93}
J. Zimba and R. Penrose,
On Bell non-locality without probabilities: More curious geometry,
\href{https://doi.org/10.1016/0039-3681(93)90061-N}{Stud. Hist. Philos. Sci. A \textbf{24}, 697 (1993).}

\bibitem{CG96}
A. Cabello and G. Garc\'{\i}a-Alcaine,
Bell-Kochen-Specker theorem for any finite dimension $n \ge 3$,
\href{https://doi.org/10.1088/0305-4470/29/5/016}{J. Phys. A: Math. Gen. \textbf{29}, 1025 (1996).}

\bibitem{RRHPHH20}
R. Ramanathan, M. Rosicka, K. Horodecki, S. Pironio, M. Horodecki, and P. Horodecki,
Gadget structures in proofs of the Kochen-Specker theorem,
\href{https://doi.org/10.22331/q-2020-08-14-308}{Quantum \textbf{4}, 308 (2020).}

\bibitem{CPSS18}
A. Cabello, J. R. Portillo, A. Sol\'{\i}s, and K. Svozil,
Minimal true-implies-false and true-implies-true sets of propositions in noncontextual hidden variable theories,
\href{https://doi.org/10.1103/PhysRevA.98.012106}{Phys. Rev. A \textbf{98}, 012106 (2018).}

\bibitem{XCG20}
Z.-P. Xu, J.-L. Chen, and O. G\"uhne,
Proof of the Peres Conjecture for Contextuality,
\href{https://doi.org/10.1103/PhysRevLett.124.230401}{Phys. Rev. Lett. \textbf{124}, 230401 (2020).}

\bibitem{BBC12}
I. Bengtsson, K. Blanchfield, and A. Cabello,
A Kochen-Specker inequality from a SIC,
\href{http://doi.org/10.1016/j.physleta.2011.12.011}{Phys. Lett. A \textbf{376}, 374 (2012).}

\bibitem{XCS15}
Z.-P. Xu, J.-L. Chen, and H.-Y. Su,
State-independent contextuality sets for a qutrit,
\href{https://doi.org/10.1016/j.physleta.2015.04.024}{Phys. Lett. A \textbf{379}, 1868 (2015).}

%%%%%%%%%%%%%%%%%%%%%%%%%%%%%%%%%%%%%%%%%%%%%%%%%%%%%%%%%%%%%%%%%%%
% Here is where the Supplemental Material and the references that are only there are mentioned in the published version. In this version, they are in different order: the order in which tehy appear in the main text
%%%%%%%%%%%%%%%%%%%%%%%%%%%%%%%%%%%%%%%%%%%%%%%%%%%%%%%%%%%%%%%%%%%

\bibitem{Cabello11}
A. Cabello,
State-independent quantum contextuality and maximum nonlocality,
\href{https://arxiv.org/abs/1112.5149}{\eprint{arXiv:1112.5149}.}

\bibitem{RH14}
R. Ramanathan and P. Horodecki,
Necessary and Sufficient Condition for State-Independent Contextual Measurement Scenarios,
\href{https://doi.org/10.1103/PhysRevLett.112.040404}{Phys. Rev. Lett. \textbf{112}, 040404 (2014).}

\bibitem{CKB15}
A. Cabello, M. Kleinmann, and C. Budroni,
Necessary and Sufficient Condition for Quantum State-Independent Contextuality,
\href{https://doi.org/10.1103/PhysRevLett.114.250402}{Phys. Rev. Lett. \textbf{114}, 250402 (2015).}

\bibitem{CKP16}
A. Cabello, M. Kleinmann, and J. R. Portillo,
Quantum state-independent contextuality requires 13 rays,
\href{https://doi.org/10.1088/1751-8113/49/38/38LT01}{J. Phys. A: Math. Theor. \textbf{49}, 38LT01 (2016).}

\bibitem{Cabello16}
A. Cabello,
Simple method for experimentally testing any form of quantum contextuality,
\href{https://doi.org/10.1103/PhysRevA.93.032102}{Phys. Rev. A \textbf{93}, 032102 (2016).}

\bibitem{CEG96}
A. Cabello, J. M. Estebaranz, and G. Garc\'{\i}a-Alcaine,
Bell-Kochen-Specker theorem: A proof with 18 vectors,
\href{https://doi.org/10.1016/0375-9601(96)00134-X}{Phys. Lett. A \textbf{212}, 183 (1996).}

\bibitem{LBPC14}
P. Lison\v{e}k, P. Badzi\c{a}g, J. R. Portillo, and A. Cabello,
Kochen-Specker set with seven contexts,
\href{https://doi.org/10.1103/PhysRevA.89.042101}{Phys. Rev. A \textbf{89}, 042101 (2014).}

\bibitem{CBTB13}
A. Cabello, P. Badzi\c{a}g, M. Terra Cunha, and M. Bourennane,
Simple Hardy-Like Proof of Quantum Contextuality,
\href{https://doi.org/10.1103/PhysRevLett.111.180404}{Phys. Rev. Lett. \textbf{111}, 180404 (2013).}

\bibitem{MANCB14}
B. Marques, J. Ahrens, M. Nawareg, A. Cabello, and M. Bourennane,
Experimental Observation of Hardy-Like Quantum Contextuality,
\href{https://doi.org/10.1103/PhysRevLett.113.250403}{Phys. Rev. Lett. \textbf{113}, 250403 (2014).}

\bibitem{KK12}
P. Kurzy\'{n}ski and D. Kaszlikowski,
Contextuality of almost all qutrit states can be revealed with nine observables,
\href{https://doi.org/10.1103/PhysRevA.86.042125}{Phys. Rev. A \textbf{86}, 042125 (2012).}

\bibitem{KZG09}
G. Kirchmair, F. Z\"ahringer, R. Gerritsma, M. Kleinmann, O. G{\"u}hne, A. Cabello, R. Blatt, and C. F. Roos,
State-Independent Experimental Test of Quantum Contextuality,
\href{http://doi:10.1038/nature08172}{Nature (London) \textbf{460}, 494 (2009).}

\bibitem{ZUZ13}
X. Zhang, M. Um, J. Zhang, S. An, Y. Wang, D.-L. Deng, C. Shen, L.-M. Duan, and K. Kim,
State-independent experimental test of quantum contextuality with a single trapped ion,
\href{https://doi.org/10.1103/PhysRevLett.110.070401}{Phys. Rev. Lett. \textbf{110}, 070401 (2013).}

\bibitem{LMZNCAH18}
F. M. Leupold, M. Malinowski, C. Zhang, V. Negnevitsky, A. Cabello, J. Alonso, and J. P. Home,
Sustained State-Independent Quantum Contextual Correlations from a Single Ion,
\href{https://doi.org/10.1103/PhysRevLett.120.180401}{Phys. Rev. Lett. \textbf{120}, 180401 (2018).}

\bibitem{GKCLKZGR10}
O. G{\"u}hne, M. Kleinmann, A. Cabello, J.-\AA. Larsson, G. Kirchmair, F. Z\"ahringer, R. Gerritsma, and C. F. Roos,
Compatibility and noncontextuality for sequential measurements,
\href{http://10.1103/PhysRevA.81.022121}{Phys. Rev. A \textbf{81}, 022121 (2010).}

\bibitem{KCK14}
P. Kurzy\'{n}ski, A. Cabello, and D. Kaszlikowski,
Fundamental Monogamy Relation between Contextuality and Nonlocality,
\href{https://doi.org/10.1103/PhysRevLett.112.100401}{Phys. Rev. Lett. \textbf{112}, 100401 (2014).}

\bibitem{ZZLZSX16}
X. Zhan, X. Zhang, J. Li, Y. Zhang, B. C. Sanders, and P. Xue,
Realization of the Contextuality-Nonlocality Tradeoff with a Qubit-Qutrit Photon Pair,
\href{https://doi.org/10.1103/PhysRevLett.116.090401}{Phys. Rev. Lett. \textbf{116}, 090401 (2016).}

%%%%%%%%%%%%%%%%%%%%%%%%%%%%%%%%%%%%%%%%%%%%%%%%%%%%%%%%%%%%%%%%%%%
% Supplemental Material and extra references in it
%%%%%%%%%%%%%%%%%%%%%%%%%%%%%%%%%%%%%%%%%%%%%%%%%%%%%%%%%%%%%%%%%%%

%\bibitem{SM} 
%See Supplemental Material for further details on the proof that every SD-C set can be extended to a critical SI-C set and on how to identify a minimal critical SI-C set containing a given SD-C set, which includes Refs.~\cite{BCGKL21,Kleinmann14,Luders51,Neumark40,Holevo80,AG93,Vorob'yev62,CHSH69,XC19,Bell66,Cabello16,Cabello11,RH14,CKB15,CEG96,Peres93,CK09,PMMM05,Pavicic17,Bub96,Bub97,GA10,Penrose00,X21,CKP16,Toh13a,Toh13b,Hardy93,KS65,CT11,CBTB13,MANCB14,Cabello13b,ADLPBC12,LBPC14}.

\bibitem{BCGKL21}
C. Budroni, A. Cabello, O. G{\"u}hne, M. Kleinmann, and J.-{\AA}. Larsson,
Quantum contextuality,
\href{https://arxiv.org/abs/2102.13036}{\eprint{arXiv:2102.13036}.}

\bibitem{Kleinmann14}
M. Kleinmann,
Sequences of projective measurements in generalized probabilistic models,
\href{https://doi.org/10.1088/1751-8113/47/45/455304}{J. Phys. A: Math. Theor. \textbf{47}, 455304 (2014).}

\bibitem{Luders51}
G. L\"uders,
\"Uber die Zustands\"anderung durch den Me{\ss}proze{\ss},
\href{http://www.physik.uni-augsburg.de/annalen/history/historic-papers/1951_443_322-328.pdf}{Ann. Phys. (Leipzig) \textbf{8}, 322 (1951);}
Concerning the state-change due to the measurement process, \href{https://doi.org/10.1002/andp.200610207}{Ann. Phys. (Leipzig) \textbf{15}, 663 (2006)].}

\bibitem{Neumark40}
M. A. Neumark,
Self-adjoint extensions of the second kind of a symmetric operator,
Izv. Akad. Nauk S.S.S.R. [Bull. Acad. Sci. U.S.S.R.] S\'er. Mat. \textbf{4}, 53 (1940) (Russian with English summary);
%53--104
Spectral functions of a symmetric operator,
Izv. Akad. Nauk S.S.S.R. [Bull. Acad. Sci. U.S.S.R.] S\'er. Mat. \textbf{4}, 277 (1940) (Russian with English summary);
%277--318
On a representation of additive operator set functions,
C.R. (Dokl.) Acad. Sci. U.R.S.S. (N.S.) \textbf{41}, 359 (1943).
%359--361

\bibitem{Holevo80}
A. S. Holevo,
{\em Probabilistic and Statistical Aspects of Quantum Theory}
%(North-Holland, Amsterdam, 1982)
(Scuola Normale Superiore Pisa, Pisa, Italy, 2011), p.~55.
First published in Russian in 1980.

\bibitem{AG93}
N. I. Akhiezer and I. M. Glazman,
{\em Theory of Linear Operators in Hilbert Space}
(Dover, New York, 1993), Vol.\ II, p.~121.

\bibitem{Vorob'yev62}
N. N. Vorob'yev,
Consistent families of measures and their extensions, 
\href{https://doi.org/10.1137/1107014}{Theory of Probab. Appl. \textbf{7}, 147 (1962).}

\bibitem{CHSH69}
J. F. Clauser, M. A. Horne, A. Shimony, and R. A. Holt,
Proposed Experiment to TeSchst Local Hidden-Variable Theories,
\href{https://doi.org/10.1103/PhysRevLett.23.880}{Phys. Rev. Lett. \textbf{23}, 880 (1969).}

\bibitem{XC19}
Z.-P. Xu and A. Cabello,
Necessary and sufficient condition for contextuality from incompatibility,
\href{https://doi.org/10.1103/PhysRevA.99.020103}{Phys. Rev. A \textbf{99}, 020103(R) (2019).}

\bibitem{Bell66}
J. B. Bell,
On the problem of hidden variables in quantum mechanics,
\href{https://doi.org/10.1103/RevModPhys.38.447}{Rev. Mod. Phys. \textbf{38}, 447 (1966).}

\bibitem{Peres93}
A. Peres,
{\em Quantum Theory: Concepts and Methods}
(Kluwer, Dordrecht, Holland, 1993).

\bibitem{CK09}
J. H. Conway and S. Kochen, 
The strong free will theorem,
%Notices of the American Mathematical Society, 56 (2009), pp.~226-232
\href{http://www.ams.org/notices/200902/rtx090200226p.pdf}{Not. Am. Math. Soc. \textbf{56}, 226 (2009).}

\bibitem{PMMM05}
M. Pavi\v{c}i\'{c}, J.-P. Merlet, B. McKay, and N. D. Megill,
Kochen–Specker vectors,
\href{https://doi.org/10.1088/0305-4470/38/7/013}{J. Phys. A: Math. Gen. \textbf{38}, 1577 (2005).}

\bibitem{Pavicic17}
M. Pavi\v{c}i\'{c},
Arbitrarily exhaustive hypergraph generation of 4-, 6-, 8-, 16-, and 32-dimensional quantum contextual sets, 
\href{https://doi.org/10.1103/PhysRevA.95.062121}{Phys. Rev. A \textbf{95}, 062121 (2017).}

\bibitem{Bub96}
J. Bub, 
Sch\"utte's tautology and the Kochen-Specker theorem,
\href{https://doi.org/10.1007/BF02058633}{Found. Phys. \textbf{26}, 787 (1996).}
%proposed by Kurt Sch\"tte in an unpublished letter to Specker in 1965.

\bibitem{Bub97}
J. Bub,
{\em Interpreting the Quantum World}
(Cambridge University Press, New York, 1997), p.~82. 

\bibitem{GA10}
E. Gould and P. K. Aravind,
Isomorphism between the Peres and Penrose proofs of the BKS~theorem in three dimensions, 
\href{https://doi.org/10.1007/s10701-010-9434-2}{Found. Phys. \textbf{40}, 1096 (2010).}

\bibitem{Penrose00}
R. Penrose, 
On Bell non-locality without probabilities: some curious geometry,
in {\em Quantum Reflections}, edited by J. Ellis and D. Amati (Cambridge University Press, Cambridge, England, 2000), p.~1.

\bibitem{X21}
Z.-P. Xu {\em et al.} (to be published).

\bibitem{Toh13a}
S. P. Toh, 
Kochen-Specker sets with a mixture of 16~rank-1 and 14~rank-2 projectors for a three-qubit system,
\href{https://doi.org/10.1088/0256-307X/30/10/100302}{Chin. Phys. Lett. \textbf{30}, 100302 (2013).}

\bibitem{Toh13b}
S. P. Toh,
State-independent proof of Kochen-Specker theorem with thirty rank-two projectors,
\href{https://doi.org/10.1088/0256-307X/30/10/100303}{Chin. Phys. Lett. \textbf{30}, 100303 (2013).}

\bibitem{Hardy93}
L. Hardy, 
Nonlocality for Two Particles Without Inequalities for Almost All Entangled States,
\href{https://doi.org/10.1103/PhysRevLett.71.1665}{Phys. Rev. Lett. \textbf{71}, 1665 (1993).}

\bibitem{KS65}
S. Kochen and E. P. Specker,
in {\it Symposium on the Theory of Models}, edited by J. W. Addison, L. Henkin, and A. Tarski
(North-Holland, Amsterdam, Holland, 1965), p.~177.

\bibitem{CT11}
A. Cabello and M. Terra Cunha,
Proposal of a Two-Qutrit Contextuality Test Free of the Finite Precision and Compatibility Loopholes,
\href{https://doi.org/10.1103/PhysRevLett.106.190401}{Phys. Rev. Lett. \textbf{106}, 190401 (2011).}

\bibitem{Cabello13b}
A. Cabello,
Twin inequality for fully contextual quantum correlations,
\href{https://doi.org/10.1103/PhysRevA.87.010104}{Phys. Rev. A \textbf{87}, 010104(R) (2013).}

\bibitem{ADLPBC12}
E. Amselem, L. E. Danielsen, A. J. López-Tarrida, J. R. Portillo, M. Bourennane, and A. Cabello,
Experimental Fully Contextual Correlations,
\href{https://doi.org/10.1103/PhysRevLett.108.200405}{Phys. Rev. Lett. \textbf{108}, 200405 (2012).}

%%%%%%%%%%%%%%%%%%%%%%%%%%%%%%%%%%%%%%%%%%%%%%%%%%%%%%%%%%%%%%%%%%%
% End of the extra references in the Supplemental Material
%%%%%%%%%%%%%%%%%%%%%%%%%%%%%%%%%%%%%%%%%%%%%%%%%%%%%%%%%%%%%%%%%%%

%%%%%%%%%%%%%%%%%%%%%%%%%%%%%%%%%%%%%%%%%%%%%%%%%%%%%%%%%%%%%%%%%%%

\end{thebibliography}
\end{document}